\documentclass[12pt,conference,onecolumn]{IEEEtran}
\usepackage{amssymb}
\usepackage{amsthm}
\usepackage{amsmath}
\usepackage{amsfonts}
\usepackage{graphicx}
\usepackage{graphics}
\usepackage{subfigure}
\usepackage{multirow}
\usepackage{setspace}
\newcommand{\limi}[1]{\;\dot{#1}\;}

\newtheorem{thm}{\it{Theorem}}
\newtheorem{cor}{\it{Corollary}}
\newtheorem{lem}{\it{Lemma}}

\newtheorem{ex}{\bf{Example}}
\title{Lattice Sequential Decoding for LAST Coded MIMO Channels: Achievable Rate, DMT, and Complexity Analysis}

\author{
\authorblockN{Walid Abediseid and Mohamed Oussama Damen}
\authorblockA{Department of Elect. \& Comp. Engineering\\
University of Waterloo \\
Waterloo, Ontario N2L 3G1 \\
\{wabedise, modamen\}@uwaterloo.ca}
}
\begin{document}
% make the title area
\maketitle
\begin{abstract}
In this paper, the asymptotic performance of the lattice sequential decoder for LAttice Space-Time (LAST) coded MIMO channel is analyzed. We determine the rates achievable by lattice coding and sequential decoding applied to such a channel. The diversity-multiplexing tradeoff (DMT) under lattice sequential decoding is derived as a function of its parameter---\textit{the bias term}, which is critical for controlling the amount of computations required at the decoding stage. Achieving low decoding complexity requires increasing the value of the bias term. However, this is done at the expense of losing the optimal tradeoff of the channel. In this work, we derive the tail distribution of the decoder's computational complexity in the high signal-to-noise ratio regime. Our analysis reveals that the tail distribution of such a low complexity decoder is dominated by the outage probability of the channel for the underlying coding scheme. Also, the tail exponent of the complexity distribution is shown to be equivalent to the DMT achieved by lattice coding and lattice sequential decoding schemes. We derive the asymptotic average complexity of the sequential decoder as a function of the system parameters. In particular, we show that there exists a \textit{cut-off} multiplexing gain for which the average computational complexity of the decoder remains bounded.
\end{abstract}
\section{Introduction}
The most important parameters for the information data transmission problem are: the rate $R$, the probability of decoding error $P_e$, the block-length $m$, and the complexity (encoding and decoding). It is well-known that the noise introduced by the channel sets up a fundamental constant on how much data can be transmitted through the channel. This constant is called the \textit{capacity} \cite{Shannon1}. For a fixed channel, our goal is to transmit at rates close to capacity with low probability of decoding error using simple coding and decoding algorithms and short codes. However, there is a tradeoff. The tradeoff between the performance, the achievable rate, and the complexity is fundamental that exits in any communication system.

Low complexity capacity-achieving codes exist. A special type of these codes which will be considered intensively in this work are constructed based on \textit{lattices}---a mathematical approach for representing infinite discrete points in the Euclidean space \cite{Conway}. The theory of lattices has become a powerful tool to analyze many point-to-point digital and wireless communication systems, particularly, communication systems that can be well-described by the \textit{linear Gaussian vector channel} model. The channel model may be mathematically expressed as
\begin{equation}\label{linear_model}
\pmb{y}=\pmb{Bx}+\pmb{e},
\end{equation}
where $\pmb{x}\in\mathbb{R}^m$ is the input to the channel, $\pmb{y}\in\mathbb{R}^m$ is the output of the channel, $\pmb{e}\in\mathbb{R}^m$ is the additive Gaussian noise vector with entries that are independent identically distributed, zero-mean Gaussian random variables with variance $\sigma^2$, i.e., $\pmb{e}\sim\mathcal{N}(\pmb{0},\sigma^2\pmb{I}_m)$, and $\pmb{B}\in\mathbb{R}^{m\times m}$ is a matrix representing the channel linear mapping. 

It is the linearity of the channel that makes it a good match to the linearity of lattices. In this paper, we assume that $\pmb{x}$ is a codeword selected uniformly from a lattice code. Let $\Lambda_c\stackrel{\Delta}{=}\Lambda(\pmb{G})=\{\pmb{x}=\pmb{G}\pmb{z}:\pmb{z}\in\mathbb{Z}^m\}$ be a lattice in $\mathbb{R}^m$ where $\pmb{G}$ is an $m\times m$ full-rank lattice generator matrix. The Voronoi cell, $\mathcal{V}_{\pmb{x}}(\pmb{G})$, that corresponds to the lattice point $\pmb{x}\in\Lambda_c$ is the set of points in $\mathbb{R}^m$ closest to $\pmb{x}$ than to any other point $\pmb{\lambda}\in\Lambda_c$, with volume that is given by $V_c\triangleq {\rm Vol}(\mathcal{V}_{\pmb{x}}(\pmb{G}))=\sqrt{{\rm det}(\pmb{G}^{\mathsf{T}}\pmb{G})}$. An $m$-dimensional lattice code $\mathcal{C}(\Lambda_c,\pmb{u}_o,\mathcal{R})$ is the finite subset of the lattice translate $\Lambda_c+\pmb{u}_0$ inside the shaping region $\mathcal{R}$, i.e., $\mathcal{C}=\{\Lambda_c+\pmb{u}_0\}\cap\mathcal{R}$, where $\mathcal{R}$ is a bounded measurable region\footnote{In this paper, we consider a shaping region $\mathcal{R}$ that corresponds to the Voronoi cell $\mathcal{V}_s$ of a sublattice $\Lambda_s$ of $\Lambda_c$, i.e., $\Lambda_s\subseteq \Lambda_c$. The generated codes are called nested (or Voronoi) lattice codes (see \cite{EZ2} for more details).} of $\mathbb{R}^m$.

For the above channel model and assuming $\pmb{B}$ is perfectly known at the receiver, it is well-known that the maximum-likelihood (ML) decoder
\begin{equation}\label{ML_decoder}
\hat{\pmb{x}}={\rm arg}\min_{\pmb{x}\in\mathcal{C}}|\pmb{y}-\pmb{B}\pmb{x}|^2,
\end{equation}
is the \textit{optimal} solution that minimizes the word error probability $P_e=\Pr(\hat{\pmb{x}}\neq\pmb{x})$, where $\hat{\pmb{x}}$ is the output of the decoder. In such a decoder, the received signal is decoded to the nearest codeword or lattice point inside $\mathcal{R}$. Lattice coding and ML decoding achieve the capacity of the channel \cite{UR}. However, in ML decoding, searching over the codebook $\mathcal{C}$ is performed by a search algorithm (e.g., the sphere decoder) that takes into account the shaping region $\mathcal{R}$ which is referred to as \textit{boundary control}. Due to its exponential complexity, the implementation of the ML decoder is practically unfeasible and the design of low complexity receivers that achieve near optimal performance is considered a challenging problem.

Relaxing the boundary control, or \textit{lattice decoding}, is believed to reduce the complexity at the expense of introducing some error performance degradation. Lattice decoder algorithms reduce the complexity by relaxing the code boundary constraint and find the point of the underlying (infinite) lattice closest to the received point (which may or may not be a
code point). In lattice theory, this is usually referred to as the closest lattice point search problem (CLPS) \cite{CLSP}, which can be described by
\begin{equation}\label{lattice_decoding}
\hat{\pmb{x}}={\rm arg}\min_{\pmb{x}\in\Lambda_c}|\pmb{y}-\pmb{B}\pmb{x}|^2.
\end{equation}
%,~\cite{Polytrev},~\cite{Loe},~\cite{MIMOARQ},~\cite{GCD}~\cite{EZ1}

Many researchers have studied the information-theoretic limits of lattice coding and decoding schemes for the linear Gaussian vector channel model (see ~\cite{deBuda} and references therein). For many scenarios that fall into such class, it has been shown that lattice decoding by itself cannot achieve the capacity of the channel at \textit{any} signal-to-noise ratio (SNR). One common way that is used to overcome such deficiency of lattice decoding is through the use of minimum mean square-error decision feedback equalization\footnote{For the case of additive white Gaussian noise channel, the capacity is achieved through the use of the linear minimum mean-square error estimator of the channel input from the channel output \cite{Forney_MMSE}.} (MMSE-DFE) \cite{GCD},\cite{EZ1}. In this case, the above channel model still applies with $\pmb{B}$ representing the feedback filter matrix of the MMSE-DFE\footnote{For finite system dimensionality, the additive noise $\pmb{e}$, although can be shown to have uncorrelated elements, may not be Gaussian. However, this have no significant effect on the results at the signal-to-noise ratio of interest (see \cite{GCD} for more details about this topic).}. We review the achievable rates of some important channels under MMSE-DFE lattice decoding that provide the main motivation to the rest of the work.
\begin{itemize}
\item \textbf{The Lattice Coded AWGN Channel}: This model corresponds to the case where the noise variance $\sigma^2=1$, and $\pmb{B}=\sqrt{1+\rho}\pmb{I}_m$, where $\rho$ is defined as the signal-to-noise ratio (SNR) at the receiver. Now, one can show that reliable communication can be made possible as long as we operate at rates
    $$R< \log\det(\pmb{B}^\mathsf{T}\pmb{B})^{1/m}= \log(1+\rho),$$
which is the capacity of the AWGN channel. A very interesting approach that may be used to prove the rate achievability of lattice coding and decoding schemes for such a channel is through the so-called the \textit{ambiguity decoder}. Lattice ambiguity decoder was originally developed by Loeliger in \cite{Loe} for the additive white Gaussian noise (AWGN) channel. The result was extended in \cite{GCD} to the quasi-static, Rayleigh fading $M\times N$ LAttice Space-Time (LAST) coded multiple-input multiple-output (MIMO) channel. The same technique will be used in this work to analyze the achievable rate of other efficient lattice decoders for the quasi-static MIMO channel.

\item \textbf{The LAST Coded $M\times N$ MIMO Channel}: We consider a quasi-static, Rayleigh fading MIMO channel with $M$-transmit, $N$-receive antennas, and no channel state information (CSI) at the transmitter and perfect CSI at the receiver. The complex base-band model of the received signal can be mathematically described by
\begin{equation}\label{complex_model}
\pmb{Y}^c=\sqrt{\rho}\pmb{H}^c\pmb{X}^c+\pmb{W}^c,
\end{equation}
where $\pmb{X}^c\in\mathbb{C}^{M\times T}$ is the transmitted space-time code matrix, $T$ is the codeword length (the number of channel usages), $\pmb{Y}^c\in\mathbb{C}^{N\times T}$ is the received signal matrix, $\pmb{W}^c\in\mathbb{C}^{N\times T}$ is the noise matrix, $\pmb{H}^c\in\mathbb{C}^{N\times M}$ is the channel matrix, and $\rho=\mathsf{SNR}/M$ is the normalized SNR at each receive antenna with respect to $M$. The elements of both the noise matrix and the channel fading gain matrix are assumed to be independent identically distributed zero mean circularly symmetric complex Gaussian random variables with variance $\sigma^2=1$. The equivalent real channel model can be described as in (\ref{linear_model}) with $\pmb{e}\sim\mathcal{N}(\pmb{0},0.5\pmb{I}_m)$.

An $M\times T$ space-time coding scheme is a full-dimensional LAST code if its vectorized (real) codebook (corresponding to the channel model (\ref{linear_model})) is a lattice code with dimension $m=2MT$. As discussed in \cite{GCD}, the design of space-time signals reduces to the construction of a codebook $\mathcal{C}\subseteq\mathbb{R}^{2MT}$ with code rate $R={1\over T}\log|\mathcal{C}|$, satisfying the input averaging power constraint
\begin{equation}\label{PC}
{1\over|\mathcal{C}|}\sum_{\pmb{x}\in\mathcal{C}}|\pmb{x}|^2\leq MT.
\end{equation}
For a fixed non-random channel matrix $\pmb{H}^c$, it has been shown in \cite{GCD} that the rate
\begin{equation}\label{LAST_rate}
R_{\rm LAST}(\rho,\pmb{H}^c)= \log\det(\pmb{B}^\mathsf{T}\pmb{B})^{1/2T}=\log\det\left(\pmb{I}_M+{\rho}(\pmb{H}^c)^\mathsf{H}\pmb{H}^c\right),
\end{equation}
 is achievable under lattice coding and MMSE-DFE lattice decoding.
\end{itemize}

Unfortunately, operating at a rate equal to capacity in the AWGN channel or at a rate equal to $R_{\rm LAST}$ in the quasi-static MIMO channel is not possible. This is due to the fact that the ML decoder or the MMSE-DFE lattice decoder (implemented via sphere decoding algorithms), suffers from high decoding complexity especially for large signal dimensions $m$ for which high rates can be achieved. As such, allowing low computational complexity search algorithms at the decoding stage requires the use of codes with rates below capacity.

It is well-known that lattice decoders that use sphere decoding algorithms can be considered as a search in a \textit{tree} \cite{MGDC}, \cite{HV}, \cite{DGC}. Generally speaking, a sphere decoding algorithm explores the tree of all possible lattice points and uses a \textit{path metric} in order to discard paths corresponding to points outside the search sphere. As an alternative to sphere decoding algorithms, \textit{sequential decoders} comprise a set of efficient and powerful decoding techniques able to perform the tree search. These decoders can achieve \textit{near}-optimal performance without suffering the complexity of the ML or the sphere decoder for coding rates not too close to the channel capacity. In this case, it is convenient to define the decoding complexity as the total number of nodes visited by the decoder during the search.

Conventional sequential decoders (e.g., Fano and Stack algorithms \cite{Fano},\cite{Stack}) were originally constructed as an alternative to the ML decoder to decode convolutional codes transmitted via discrete memoryless channels while achieving low (average) decoding complexity. Similar to the sphere decoder, the sequential decoder uses a path metric in order to eliminant a large subset of lattice points that have low chance to be extended by the search algorithm.

For a general discrete memoryless channel, the sequential decoder's path metric, termed as the Fano metric, is given by (see \cite{JZ})
\begin{equation}\label{fano_metric}
\mu(\pmb{x}_1^k)=\log\left(\Pr(\mathcal{H}(\pmb{x}_1^k)p(\pmb{y}_1^k|\mathcal{H}(\pmb{x}_1^k)\over p(\pmb{y}_1^k)\right),
\end{equation}
where $\mathcal{H}(\pmb{x}_1^k)$ is the hypothesis that $\pmb{x}_1^k$ forms the first $k$ symbols of the transmitted sequence, $\pmb{y}_1^k$ is the first $k$ symbols of the received sequence, and $p(\cdot)$ is the probability distribution function. We review the path metric of some important channels, and discuss the rate-complexity tradeoff achieved by the decoder and the optimal metric that leads to low decoding complexity.

\begin{itemize}
\item \textbf{The Binary Symmetric Channel}: In such a channel, the sequential decoder is used to decode linear convolutional codes. For $1\leq k\leq m$, if $\Pr(\mathcal{H}(\pmb{x}_1^k))$ is uniform over all nodes $\pmb{x}_1^k$ that consist of the first $k$ components of any valid codeword in the code, the path metric can be written as
    $$\mu(\pmb{x}_1^k)=\sum\limits_{i=1}^k\log{p(y_i|x_i)\over p(y_i)} - bk,$$
    where $p(y_i|x_i)$ is the channel transition probability, and $b$ is the \textit{bias term}. The bias is introduced to favor a longer path which is closer to the end of the tree and thus is more likely to be part of the optimal code path. Massey \cite{Massey} proved that at any decoding stage, extending the path with the largest Fano metric minimizes the probability that the extending path does not belong to the optimal code path. Moreover, Messey showed that the optimal bias that minimizes the computational complexity while achieving good error performance is $b=R$, the code rate.

    Although sequential decoding algorithms are simple to describe, the analysis of the decoder's computational complexity is considered difficult. This is due to the fact that the amount of computations performed by the decoder attempting to decode a message is random. Therefore, sequential decoding complexity is usually analyzed through its \textit{computational distribution}. For codes transmitted at rate $R$, the computational complexity of the sequential decoder, denoted by $C$, for the above mentioned channel follows a Pareto distribution \cite{Jacobs},
    \begin{equation}
    \Pr(C>L)\approx L^{-e(R)},\quad L \rightarrow \infty,
    \end{equation}
    where $L$ is the distribution parameter, and $e(R)$ is the tail distribution exponent that is a function of $R$. Theoretical analysis showed that $e(R)>1$ as long as $R<R_0$, where $R_0$ is called the channel \textit{computational cut-off} rate which is strictly less than the channel capacity. This means that the average computational complexity is kept bounded as long as we operate at rates below $R_0$. Therefore, for many coding and decoding schemes, $R_0$ is considered to be the ``practical'' capacity of the channel.

\item \textbf{The Lattice Coded AWGN Channel}: The problem of detecting and decoding the received signal $\pmb{y}$ is transformed into a tree search algorithm using the QR-decomposition on the code matrix $\pmb{G}$. Let $\pmb{Q}$ and $\pmb{R}$ be the orthonormal matrix and the upper triangular matrix with positive diagonal elements, respectively, that correspond to the QR-decomposition of $\pmb{G}$. Assuming $\pmb{x}=\pmb{Gz}$ was transmitted, then one can show that for moderate-to-large SNR, the Fano matric (\ref{fano_metric}) can be expressed as (see \cite{ONM}, \cite{MGDC})
    \begin{equation}\label{mod_metric}
    \mu(\pmb{z}_1^k)=bk-|{\pmb{y}'}_1^{k}-\pmb{R}_{kk}\pmb{z}_1^k|^2,\quad \forall\;1\leq k\leq m,
    \end{equation}
    where $\pmb{z}_1^k=[z_k,\cdots,z_2,z_1]^{\mathsf{T}}$ denotes the last $k$ components of the integer vector $\pmb{z}$, $\pmb{R}_{kk}$ is the lower $k\times k$ part of the matrix $\pmb{R}$, and ${\pmb{y}'}_1^k$ is the last $k$ components of the vector ${\pmb{y}'}=\pmb{Q}^{\mathsf{T}}\pmb{y}$.

    There are several works that discuss sequential decoding for the lattice coded AWGN channel \cite{TVZ}--\cite{ONM}. For such a channel, it is well-known that sequential decoding of lattice codes can operate ``efficiently'' (with bounded average complexity) at rates below $R_0$ which is only a factor of approximately 1.7 dB away from capacity \cite{FTC}. It is important to note that for sequential decoding algorithms that approximate lattice decoding, choosing $b=R$ is not the optimal solution that minimizes the average complexity. This is due to the infinite number of \textit{virtual} codewords as seen by the decoder, and hence the rate $R$ is meaning less. Therefore, one should appropriately select the bias term $b$ so that we attain near-optimal performance while reducing the computational complexity. This may be achieved by the sequential decoder by ensuring that the metric along the correct path increases on average, while decreases along other paths. In this case, we choose $b$ such that $\mathsf{E}_{\pmb{e}}\{\mu(\pmb{x}_1^k)\}>0$ (assuming $\pmb{x}$ is the correct path). This corresponds to $b>\mathsf{E}\{|[\pmb{e}]_i|^2\}=\sigma^2$. In fact, the optimal bias under lattice coding and sequential decoding schemes that is used to minimize the average decoding effort was derived in \cite{ONM} and is given by\footnote{This optimal bias term was derived based on the work by Poltyrev in \cite{Polytrev}. In his work, Poltyrev considered the fundamental limits achieved by lattice coding and decoding for the unconstrained AWGN channel. A new notion of capacity was considered in his work based on the characteristic of lattices (the lattice density to be specific) which is termed as the volume-to-noise ratio (see \cite{Polytrev} for more details).}
    \begin{equation}\label{bias_AWGN}
    b=\sigma^2\log{4\over \pi\sigma^2},
    \end{equation}
    where $\sigma^2$ is the variance of the channel noise.

\item \textbf{The $M\times N$ MIMO Channel}: Applying sequential decoders for the detection of signals transmitted via MIMO communication channels introduced an alternative and interesting approach to solve the CLPS problem \cite{CLPS} that is related to the optimal decoding rule in such channels \cite{MGDC}, \cite{NMO}. Murugan et.~al.~\cite{MGDC} showed that lattice sequential decoders, although sub-optimal, are capable of achieving good, and for some cases near-ML, error performance. The analysis was considered \textit{only} for the case of uncoded MIMO channel (i.e., V-BLAST). It was demonstrated that lattice sequential decoders achieve the \textit{maximum} receive diversity provided by the channel, and for low signal dimensions these decoders achieve near-ML performance while significantly reducing decoding complexity compared to lattice decoders. Specifically, they showed that for any \textit{fixed} (large enough) $b$, the sequential decoder achieves the maximum diversity gain $N$ with computational complexity that scales at most linearly with the signal dimension $m$ as long as $\rho>\rho_0$, where $\rho_0$ is the minimum SNR required for the average complexity to remain bounded. It must be noted that $\rho_0$ is a function of $m$ and $b$ and increases proportionally with these parameters as they increase.

     Interestingly, the lattice sequential decoder allows for a systematic approach for trading off performance for complexity. It was argued in \cite{MGDC} that as $b\rightarrow 0$ we achieve the best (lattice decoding) performance but at the price of high complexity. It has been shown in \cite{MGDC} via simulation, that there exists a value of $b$, say $b^*$, such that for all $b\geq b^*$, the average computational complexity decreases monotonically with $b$. As $b\rightarrow \infty$, the sequential decoder becomes equivalent to the MMSE-DFE decoder and the number of visited nodes is always equal to $m$ at \textit{any} SNR. It is well-know \cite{MMSE2} that the MMSE-DFE decoder achieves a diversity gain equal to $N-M+1$. Therefore, one must expect that as we vary the bias term from 0 to $\infty$, the diversity gain must change from $N$ to $N-M+1$. How the diversity order changes with the bias parameter was not shown in \cite{MGDC}. Moreover, \textbf{the performance limits and the complexity achieved by lattice sequential decoders for (lattice) \textit{space-time coded} MIMO channel~\cite{GCD},~\cite{Tarokh1},~\cite{Hassibi} have not yet been studied}. This will be the main topic of the work presented here.

\end{itemize}

\subsection{Sequential Decoding for LAST Coded MIMO Channels}
In the quasi-static MIMO channel, achieving higher performance via \textit{diversity} and higher data rate via \textit{multiplexing} require incorporating error control coding (across antenna and time) at the transmitter (e.g., space-time codes). The Diversity-Multiplexing Tradeoff (DMT) \cite{ZT} has become the standard tool that is used to evaluate the performance limits of any coding and decoding schemes applied over outage-limited MIMO channels. Let the multiplexing gain $r$ be defined as in \cite{ZT}:
$$r=\lim_{\rho\rightarrow\infty}{R(\rho)\over\log{\rho}},$$
and the  diversity gain $d$ be defined as \cite{ZT}:
$$d=\lim_{\rho\rightarrow\infty}{-\log P_e(\rho)\over\log{\rho}}.$$
With the aid of the MMSE-DFE at the decoding stage, LAST coding and lattice decoding achieve (see \cite{GCD}) the \textit{optimal} tradeoff, denoted by $d_{\rm out}(r)$, of the channel
\begin{equation}\label{optimal_DMT}
d^*_{\rm out}(r)=(M-r)(N-r),\quad \forall 0\leq r\leq \min\{M,N\}.
\end{equation}

However, lattice decoders implemented via sphere decoding algorithms are only efficient in the high SNR regime and low signal dimensions, and exhibit exponential (average) complexity for low-to-moderate SNR and large signal dimensions \cite{HV}, \cite{JB}. On the other extreme, linear and non-linear receivers such as zero-forcing, MMSE, and MMSE-DFE decoders, are considered attractive alternatives to lattice decoders in MIMO channels and have been widely used in many practical communication systems \cite{Fosch}--\cite{MMSE2}. Unfortunately, the very low decoding complexity advantage that these decoders can provide comes at the expense of poor performance, especially for large signal dimensions. In fact, linear decoders cannot achieve the optimal DMT. However, with the aid of lattice reduction techniques \cite{YW}, Jald\'{e}n and Elia \cite{EJ} showed that the optimal tradeoff $d_{\rm out}^*(r)$ can be achieved using linear decoders at a worst-case complexity $O(\log\rho)$. This corresponds to a linear increase in complexity as a function of the code rate $R=r\log\rho$ at high SNR. Unfortunately, as mentioned in \cite{EJ}, this very low decoding complexity comes at the expense of a ``large'' performance (or SNR) \textit{gap} from the lattice decoder's error performance.

The problem of designing low complexity receivers for the MIMO channel that achieve \textit{near-optimal} performance (i.e., with improved SNR gap from ML or lattice decoding) is considered a challenging problem and has driven much research in the past years. In this work, we analyze the performance of \textit{lattice sequential decoding} that is capable of bridging the gap between lattice (or sphere) decoders and low complexity linear decoders (e.g., MMSE-DFE decoder).

The problem of detecting and decoding the received signal $\pmb{y}$ is transformed into a tree search algorithm using the QR-decomposition on the channel-code matrix $\pmb{BG}$. Similar to the lattice coded AWGN channel, the Fano metric that corresponds to the LAST coded MIMO channel is also given by (\ref{fano_metric}). However, there are few differences when it comes to the choice of the sequential decoding parameter $b$. First, as will be shown in the sequel, for the lattice coded AWGN channel, choosing the optimal value of $b$, provided in (\ref{bias_AWGN}), as the bias allows us to operate at rates close to $R_0$ (the cut-off rate). Unfortunately, the fading nature of the MIMO wireless channel prevents us from selecting an optimal value of $b$ that minimizes the decoding complexity for all SNR while still achieving near optimal performance. To further illustrate on this point, our results show that the achievable rate of the lattice sequential decoder, denoted by $R_b(\rho,\pmb{H}^c)$, applied to the LAST coded MIMO channel is equivalent to the achievable rate of the MMSE-DFE lattice decoding \textit{offset} by a term that depends solely on the bias and increases proportionally with $b$, i.e.,
$$R_b(\rho,\pmb{H}^c)=\log\det(\pmb{I}_M+\rho(\pmb{H}^c)^\mathsf{H}\pmb{H}^c)-\Gamma(b)$$
with $\Gamma(0)=0$. Now, one can see that when the channel is near outage, transmission at non-zero rates may not be possible for large values of bias term, for which low decoding complexity is expected. In order to overcome this problem, we either increase the transmission power which may not be possible due to some power constraints, or lower the value of the bias $b$ but at the price of increasing the decoding complexity.

It is clear from the above equation that, depending on the channel condition, the bias term should be adapted accordingly in order to maintain a non-zero achievable rate at any SNR. Now, if $b=0$ one achieves the optimal DMT $d_{\rm out}^*(r)$ but at the expense of increasing the decoding complexity. On the other hand, if we let $b\rightarrow\infty$, we achieve a DMT $d_{\infty}(r)=(N-M+1)(1-r/M)$, which corresponds to the tradeoff achieved by the MMSE-DFE decoder. Such a decoder has the lowest computational complexity (equal to $m$) but at the price of very poor performance. Therefore, by varying the bias term from 0 to $\infty$ we achieve different achievable DMT curves. In this paper, we will show in details how the achievable DMT changes with the bias.

\subsection{Outline of the Main Contributions}
The contribution of this paper can be classified into two classes: the asymptotic performance analysis of the lattice sequential decoder in terms of the achievable DMT, and the computational complexity of the decoder in terms of the complexity tail distribution and the average complexity. In order to fully characterize the achievable DMT of the decoder, we determine for the first time the rates achievable by lattice coding and sequential decoding applied to the outage-limited MIMO channel. We derive the DMT as a function of the decoder bias term, which is critical for controlling the amount of computations required at the decoding stage. Achieving low decoding complexity requires increasing the value of the bias term. However, this is done at the expense of losing the optimal tradeoff of the channel. In terms of performance analysis, the work establishes the DMT optimality of \textit{fixed}-bias lattice sequential decoding.

We analyze in details the computational tail distribution of the decoder and its average complexity. Specifically, we show that, at the high SNR regime, when the computational complexity exceeds a certain limit, say $L_0$, the tail distribution becomes upper bounded by the asymptotic outage probability achieved by LAST coding and sequential decoding schemes, i.e.,
$$\Pr(C\geq L)\leq \rho^{-d^*_{\rm out}(r)},\quad L\geq L_0,$$
where $d^*_{\rm out}(r)$ is the optimal DMT of the channel, and $L_0$ is a random variable that depends on the channel condition and the code matrix. This interesting result suggests that one may save on decoding complexity while still achieving near-outage performance by setting a \textit{time-out} limit at the decoder so that when the computational complexity exceeds this limit the decoder terminates the search and declares an error.

Similar to the discrete memoryless channel, our analysis reveals that, for a fixed bias sequential decoding algorithm, there exists a \textit{cut-off} multiplexing gain, denoted by $r_0$, for which the average computational complexity of the lattice sequential decoder remains bounded as long as we operate below such value. We argue that, in order to operate at multiplexing gains beyond $r_0$, large values of $b$ must be used. However, this comes at the price of loosing the optimal tradeoff. Hence, the lattice sequential decoder provides a systematic approach for tradeoff DMT, cut-off multiplexing gain, and complexity.

Our work is organized as follows. In Section II, we briefly describe the operation of various sequential decoding algorithms. In section III, we investigate the achievable rates of lattice sequential decoders for the outage-limited MIMO channel, and we derive the \textit{general} DMT achieved by the decoder as a function of its parameter --- the bias term. We show how this parameter plays a fundamental role in determining the DMT achieved by sequential decoding of lattice codes. The optimality of the lattice sequential decoder for the quasi-static MIMO channel is proven for finite bias term. The bias term is responsible for the excellent performance-complexity tradeoff achieved by the decoder. Sections IV and V provide complete analysis for the computational complexity tail distribution and the average complexity of the lattice sequential decoder in the high SNR regime. In section VI, our theoretical analysis is supported through simulation results. Finally, conclusions are provided in section VII.

Throughout the paper, we use the following notation. The superscript $^c$ denotes complex quantities, $^\mathsf{T}$ denotes transpose, and $^\mathsf{H}$ denotes Hermitian transpose. We refer to $g(z)\limi{=}z^a$ as $\lim_{z\rightarrow\infty}g(z)/\log(z)=a$, $\dot{\geq}$ and $\dot{\leq}$ are used similarly. For a bounded Jordan-measurable region $\mathcal{R}\subset\mathbb{R}^m$, $V(\mathcal{R})$ denotes the volume of $\mathcal{R}$, and $\pmb{I}_m$ denotes the $m\times m$ identity matrix. We denote $\mathcal{S}_m(r)$ by the $m$-dimensional hypersphere of radius $r$ with $V(\mathcal{S}_m(r))=(\pi r^2)^{m/2}/\Gamma(m/2+1)$, where $\Gamma(x)$ denotes the Gamma function.

\section{Lattice Fano/Stack Sequential Decoder}
The sequential search on a tree can be briefly described as follows: the search is attempted one branch at a time. Namely, if the decoder is ``located'' at a particular node, it will move forward along the most likely branch stemming from it and thus reach a new node, provided that the likelihood of the entire past path up to and including the new node exceeds a certain current threshold. If it does not, then the decoder must return to the preceding node. From there it will try to move forward along an alternate path. It will succeed in this attempt if the value of the likelihood of the new path exceeds a threshold appropriate to it. Thus, the decoder moves forward and backward with the hope that the likely paths are going to be examined so that the average decoding effort will be kept low.

Fano and Stack sequential decoders \cite{Fano}, \cite{Stack} are efficient tree search algorithms that attempt to find a ``best fit'' with the received noisy signal. As in conventional sequential decoder, to determine a best fit (path), the path metric given in (\ref{fano_metric}) is assigned to each node on the tree.

In the Stack algorithm, as the decoder searches the different nodes in the tree, an ordered list of previously examined paths of different lengths is kept in storage. Each stack entry contains a path along with its metric. Each decoding step consists of extending the top (best) path in the stack and reordering the stack list. The decoding algorithm terminates when the top path in the stack reaches the end of the tree (refer to \cite{Stack} for more details about the algorithm).

In the Fano algorithm, as the decoder searches nodes, the path metric is compared to a certain threshold denoted by $\mathcal{\tau}\in\{\cdots,-2\delta,-\delta,0,\delta,2\delta,\cdots\}$ where $\delta$ is called the step size. The decoder attempts to extend the most probable path by moving ``forward'' if the path metric stays above the running threshold. Otherwise, it moves ``backward'' searching for another path that may lead to the most probable transmitted sequence (refer to \cite{Fano} for more details about the algorithm).

Although the Stack decoder and the Fano algorithm generate essentially the same set of visited nodes (see \cite{MGDC}), the Fano decoder visits some nodes more than once. However, the Fano decoder requires essentially no memory, unlike the Stack algorithm. Also, it must be noted that the way the nodes are generated in both sequential algorithms plays an important role in reducing the computation complexity and for some cases may improve the detection performance. For example, the determination of the best and next best nodes is simplified in the CLPS problem by using the Schnorr-Euchner enumeration~\cite{DGC} which generates nodes with metrics in ascending order given any node $\pmb{z}_1^k$. However, it should be noted that for the entire paper, and for the sake of simplifying the analysis, we will consider the use of the Stack algorithm in our performance and complexity analysis. For the Fano algorithm the same results also apply.

\section{Outage Performance Analysis}
Our goal in this section is to analyze the DMT achieved by the LAST coding and MMSE-DFE lattice sequential decoding applied to the quasi-static $M\times N$ MIMO channel. The achievable DMT of a particular coding and decoding schemes in such a channel is usually derived using the outage probability $P_{\rm out}(\rho,R)$ which is defined as the probability that the coding rate cannot be supported by the channel. In other words, an outage occurs if the coding rate $R$ exceeds the achievable rate of the channel. As such, determining the achievable rate under LAST coding and sequential decoding is essential in order to determine the DMT.

\subsection{Achievable Rate}
As discussed in the introduction, the sequential decoder's output depends critically on the bias term $b$ (defined in (\ref{mod_metric})). Therefore, it is to be expected that the achievable rate as well as the outage probability will depend heavily on such decoding parameter. As discussed in the previous section, rates up to $\log\det(\pmb{I}_M+\rho(\pmb{H}^c)^\mathsf{H}\pmb{H}^c)$ are achievable by lattice coding and decoding. When the lattice decoder is replaced by the lattice Fano/Stack\footnote{For the Fano algorithm, we assume throughout the paper that only small values of step size $\delta$ is used by the decoder, and hence, its affect on the performance analysis can be neglected (see the proof of Theorem 4). Otherwise, choosing very large values of $\delta$ may result in very poor performance. For the Stack algorithm, we have $\delta=0$.} sequential decoder we get the following result:
\begin{thm}
For a fixed non-random channel matrix $\pmb{H}^c$, the rate
\begin{equation}\label{fano_rate}
R_b(\pmb{H}^c,\rho)\triangleq \max\biggl\{\log\det(\pmb{I}_M+\rho(\pmb{H}^c)^\mathsf{H}\pmb{H}^c)-2M\log\left({1+\sqrt{1+8\alpha}\over 2}\right),0\biggr\},
 \end{equation}
is achievable by LAST coding and MMSE-DFE lattice sequential decoding with bias term $b$, where $\alpha$ is given by
\begin{equation}\label{alpha}
\alpha=\left(r_{\rm eff}(\pmb{B}\pmb{G})\over 2r_{\rm pack}(\pmb{B}\pmb{G})\right)^2b,
\end{equation}
and $r_{\rm eff}(\pmb{B}\pmb{G})$ and $r_{\rm pack}(\pmb{B}\pmb{G})$ are the effective radius and packing radius\footnote{The effective radius of a lattice $r_{\rm eff}$ is defined as the radius of the sphere with volume equal to the volume of the fundamental (Voronoi) region of the lattice. The packing radius $r_{\rm pack}$ is the largest radius of the sphere that is contained inside the Voronoi region of the lattice.} of the lattice generated using $\pmb{B}\pmb{G}$, respectively.
\end{thm}

It should be noted that the above theorem applies to the general linear Gaussian vector channel model that is described in (\ref{linear_model}) with arbitrary $\pmb{B}$. As an example, consider again the lattice coded AWGN (non-fading) channel with $\pmb{B}=\sqrt{1+\rho}\pmb{I}_m$ under the use of MMSE-DFE lattice sequential decoding. As $m\rightarrow\infty$, for a well-constructed lattice code ensamble\footnote{Codes that are constructed using lattices that satisfy the Minkowski-Hlawka theorem (see \cite{Loe}--\cite{Polytrev} for more details)}, it is well-known that\footnote{As discussed in \cite{EZ2}, the packing efficiency $\eta_{\rm pack}$, defined as ${r_{\rm pack}(\pmb{G})/ r_{\rm eff}(\pmb{G})}$, of a well-constructed lattice ensemble $\Lambda(\pmb{G})$ is asymptotically (as $m\rightarrow\infty$) bounded by $0.5\leq \eta_{\rm pack}\leq 0.66$.}
$${2r_{\rm pack}(\pmb{B}\pmb{G})\over r_{\rm eff}(\pmb{B}\pmb{G})}\approx 1  $$
In this case, the achievable rate in (\ref{fano_rate}) reduces simply to
$$R_b(\rho)\approx\log(1+\rho)-2\log\left({1+\sqrt{1+8b}\over 2}\right).$$
As discussed in the introduction, the optimal bias that achieves the best performance while maintaining low complexity decoding was found in \cite{ONM} and is given in (\ref{bias_AWGN}). Substituting this value of $b$ in the above equation with channel noise variance $\sigma^2=0.5$, we get $R_b(\rho)\approx \log(1+\rho)-1.64$, which is approximately 2 dB away from capacity. The achievable rate $R_b$ in this case is very close to the computational cut-off rate $R_0$ of the AWGN channel.

Some special remarks can be made about Theorem 1. First, it is clear that as $b\rightarrow 0$ the achievable rate $R_b\rightarrow R_{\rm LAST}=\log\det(\pmb{I}_M+\rho(\pmb{H}^c)^\mathsf{H}\pmb{H}^c)$. In fact, as will be shown in the sequel, as $b\rightarrow 0$ we achieve the lattice decoder error performance at \textit{any} SNR. As $b$ becomes large, it may not be possible to transmit data at rates close to $R_{\rm LAST}$. The problem that we may encounter here is that when the channel is near outage, for a fixed value of $b$ we may not be able to send data at a positive rate, especially at low-to-moderate SNR. In this case, it is highly likely that the decoder will perform erroneous detection. However, as $\rho\rightarrow \infty$, for \textit{any} fixed $b$, the offset term that appears in the achievable rate equation becomes negligible and at very high SNR we have $R_b\approx R_{\rm LAST}$. Therefore, for a fixed bias, although we may not achieve close to lattice decoding performance, one should expect the sequential decoder to achieve the optimal DMT of the channel $d^*_{\rm out}(r)$. This is summarized in the following theorem:

\begin{thm}
There exists a sequence of nested LAST codes with block length $T\geq M+N-1$ that achieves the optimal DMT curve $d^*_{\rm out}(r)=(M-r)(N-r)$ for all $r\in [0,\min\{M,N\}]$ under LAST coding and MMSE-DFE lattice sequential decoding for any fixed bias $b\geq0$.
\end{thm}
\begin{proof}
 (Sketch) For a fixed bias, we have $R_b\approx R_{\rm LAST} - \gamma$ where $\gamma$ is a constant that depends on $b$. Let $R(\rho)=r\log\rho$ where $r$ is the multiplexing gain, and denote $0\leq\lambda_1\leq\cdots\leq\lambda_M$ the eigenvalues of $(\pmb{H}^c)^\mathsf{H}{\pmb{H}^c}$. Let $\alpha_i\triangleq-\log\lambda_i/\log\rho$. Using the definition of the outage probability in \cite{ZT}, we have
\begin{equation}
\begin{split}
P_{\rm out}(\rho,R)=\Pr(R(\rho)>R_b(\rho))&=\Pr(r\log\rho> R_{\rm LAST}(\rho)-\gamma)\\
&\limi{=}\Pr\left(r\log\rho\geq \sum\limits_{i=1}^M(1-\alpha_i)^+\log\rho-\gamma\right)\\
                                          &=\Pr\left(r> \sum\limits_{i=1}^M(1-\alpha_i)^+-{\gamma\over \log\rho}\right)\\
                                          &\limi{=} \Pr\left(r> \sum\limits_{i=1}^M(1-\alpha_i)^+\right)
                                                                                    \end{split}
\end{equation}
The last equation represents the achievable outage probability under ML (or MMSE-DFE lattice) decoding (see \cite{ZT} or \cite{GCD}). Therefore, at high SNR we have $P_{\rm out}(\rho,r\log\rho)\limi{=}\rho^{-d_{\rm out}^*(r)}$. According to \cite{ZT}, the outage probability serves as a lower bound for the probability of decoding error. Therefore, we have
\begin{equation}
P_e(\rho)\limi{\geq} \rho^{-d_{\rm out}^*(r)}.
\end{equation}
The proof is completed if we show that $P_e(\rho)$ is asymptotically upper bounded by $\rho^{-d_{\rm out}^*(r)}$. This is done in Appendix I.
\end{proof}

The above theorem indicates that the use of the ML or near-ML receivers (e.g., lattice decoders) is not essential if the main goal is to achieve the optimal tradeoff of the channel. Sub-optimal receivers may do the job. This result agrees with the work by Jald\'{e}n and Elia \cite{EJ} where they considered the use of lattice aided-reduction linear decoders and proved their optimality in the DMT sense. These decoders, although achieve very low decoding complexity, suffers from a large SNR gap from the ML or lattice decoding error performance. However, lattice sequential decoders allow for a \textit{systematic} approach for trading off performance for complexity. Using a fixed but large value of $b$, although achieves the optimal DMT, the performance (SNR) gap from the ML or the lattice decoder increases as $b$ becomes large. To achieve near-ML performance in this case, one has to resort to \textit{low} fixed values of $b$.

A natural question that may be asked at this point is: how large $b$ can be set in order not to loose the optimal tradeoff? For fixed (finite) $b$, one cannot catch, in general, the effect of the bias term on the DMT achieved by such decoding scheme. As will be shown in the sequel, in order to do that, we need to allow the bias term to vary according to the channel condition and the SNR.

Now, before proving Theorem 1, we would like to introduce the so called \textit{ambiguity decoder}. The lattice ambiguity decoder was originally developed by Loeliger in \cite{Loe}, and was used in \cite{GCD} to prove the achievability rate of the MMSE-DFE lattice decoder that is given in (\ref{LAST_rate}). The same technique will be used in this paper to derive the achievable rate under MMSE-DFE lattice sequential decoding.

Assume the received vector can be written as $\pmb{y}=\pmb{x}+\pmb{w}$, where $\pmb{x}\in\Lambda_c$ and $\pmb{w}=\pmb{A}^{-1}\pmb{e}$ is an $m$-dimensional noise vector independent of $\pmb{x}$, for which $\pmb{A}\in\mathbb{R}^{m\times m}$ is an arbitrary full-rank matrix and $\pmb{e}\sim\mathcal{N}(\pmb{0},0.5\pmb{I})$. The ambiguity decoder is defined by a decision region $\mathcal{E}\subset\mathbb{R}^m$ and outputs ${\pmb{x}}\in\Lambda_c$ if $\pmb{y}\in\mathcal{E}+\pmb{x}$ and there exists no other point $\pmb{x}'\in\Lambda_c$ such that $\pmb{y}\in\mathcal{E}+\pmb{x}'$. An ambiguity occurs if the received vector $\pmb{y}\in\{\mathcal{E}+{\pmb{x}}\}\cap\{\mathcal{E}+\pmb{x}'\}$ for some $\pmb{x}\neq\pmb{x}'$. If we define  $\mathcal{A}(\mathcal{E})$ to be the ambiguity event for the decision region $\mathcal{E}$, then for a given $\Lambda_c$ and $\mathcal{E}$, the probability of error can be upper bounded  as
\begin{equation}\label{AMBUB}
P_e(\mathcal{E}|\Lambda_c)\leq {\rm Pr}(\pmb{e}\notin \mathcal{E}) + {\rm Pr}(\mathcal{A}(\mathcal{E})).
\end{equation}
As mentioned in \cite{Loe}, the upper bound (\ref{AMBUB}) holds for any Jordan measurable bounded subset $\mathcal{E}$ of $\mathbb{R}^m$.
Consider now the following lemma:
\begin{lem}
There exists an $m=2MT$-dimensional lattice code $\mathcal{C}(\Lambda_c,\pmb{u}_0,\mathcal{R})$ with fundamental volume $V_c$ that satisfies (\ref{PC}), for some fixed translation vector $\pmb{u}_0$, and $\mathcal{R}$ is the $m/2$-dimensional hypersphere with radius $\sqrt{MT}$ centered at the origin such that the error probability is upper bounded as
\begin{equation}\label{Pe1}
P_e(\Lambda_c,\mathcal{E}_{T,\gamma})\leq (1+\epsilon')2^{-T[\log\det(\pmb{A}^\mathsf{T}\pmb{A})^{1/2T}-M\log(2r_e^2/m)-R]} +{\rm Pr}(\pmb{e}\notin \mathcal{E}_{T,\gamma}),
\end{equation}
where $\mathcal{E}_{T,\gamma}\triangleq\{\pmb{z}\in\mathbb{R}^{2MT}:\pmb{z}^\mathsf{T}\pmb{A}^{\mathsf{T}}\pmb{A}\pmb{z}\leq r_e^2(1+\gamma)\}$, $r_e>0$, $\gamma>0$, and  $\epsilon'>0$.
\end{lem}
\begin{proof}
See \cite{GCD}.
\end{proof}

The achievable rate under MMSE-DFE lattice decoding provided in (\ref{LAST_rate}) follows easily by letting $\pmb{A}=\pmb{B}$ and $r_e^2=MT$ in the above lemma. In that case, from the standard typicality arguments it follows that for any $\epsilon>0$ and $\gamma>0$, there exists $T_{\gamma,\epsilon}$ such that for all $T>T_{\gamma,\epsilon}$ we have that ${\rm Pr}(\pmb{e}\notin \mathcal{E}_{T,\gamma})<\epsilon/2$. The second term in the upper bound (\ref{Pe1}) can be made smaller than $\epsilon/2$ for sufficiently large $T$ if $R<R_{\rm LAST}$.

\subsection{Proof of Theorem 1}
\begin{proof}
Consider an $m=2MT$-dimensional lattice code $\mathcal{C}(\Lambda_c,\pmb{u}_0,\mathcal{R})$ (that corresponds to a generated matrix $\pmb{G}$) with fundamental volume $V_c$ that satisfies (\ref{PC}), for some fixed translation vector $\pmb{u}_0$, and $\mathcal{R}$ is the $m/2$-dimensional hypersphere with radius $\sqrt{m/2}$ centered at the origin.

The input to the MMSE-DFE lattice sequential decoder is the vector $\pmb{y}'=\pmb{Q}^\mathsf{T}{\pmb{y}}$, where $\pmb{Q}$ is an orthogonal matrix that corresponds to the QR decomposition of the channel-code matrix $\pmb{B}\pmb{G}=\pmb{Q}\pmb{R}$. The associated path metric in this case is given by (\ref{mod_metric}).

Consider the Stack algorithm with bias $b\geq 0$. Let $E_s$ be the event that the Stack decoder makes an erroneous detection. Due to lattice symmetry, we can assume that the all zero lattice point, i.e., $\pmb{0}$, was transmitted. For a given lattice $\Lambda_c$, the frame error rate of the lattice Stack sequential decoder.
 \begin{equation}\label{UB200}
 \begin{split}
\Pr(E_s|\Lambda_c)&\stackrel{(a)}{\leq}{\rm Pr}\left(\bigcup_{\pmb{z}\in\mathbb{Z}^m\backslash\{\pmb{0}\}}\{\mu(\pmb{z})>\mu_{\min}\}\right)\\
&\stackrel{(b)}{\leq}{\rm Pr}\left(\bigcup_{\pmb{x}\in\Lambda_c^*}\{|\pmb{Bx}|^2-2(\pmb{Bx})^{\mathsf{T}}\pmb{e}<bm\}\right)\\
&={\rm Pr}\left(\bigcup_{\pmb{x}\in\Lambda_c^*}\left\{{2(\pmb{Bx})^\mathsf{T}\pmb{e}}\geq{|\pmb{Bx}|^2}\left(1-{bm\over |\pmb{Bx}|^2}\right)\right\}\right),
\end{split}
\end{equation}
where $\Lambda_c^*=\Lambda_c\backslash\{\pmb{0}\}$, $\mu_{\min}=\min\{0,b-|{\pmb{e}'}_1^1|^2,2b-|{\pmb{e}'}_1^2|^2,\ldots,bm-|{\pmb{e}'}_1^m|^2\}$ is the minimum metric that corresponds to the transmitted path, $\pmb{e}'=\pmb{Q}^\mathsf{T}\pmb{e}$, $(a)$ is due to the fact that in general, $\mu(\pmb{z})>\mu_{\min}$ is just a necessary condition for $\pmb{x}=\pmb{G}\pmb{z}$ to be decoded by the Fano decoder, and $(b)$ follows by noticing that $-(\mu_{\min}+|{\pmb{e}'}|^2)\leq 0$. It is clear from the above analysis that lattice Stack sequential decoder approaches the performance of lattice decoder as $b\rightarrow 0$. We make use of the fact that
$$|\pmb{Bx}|^2\geq \min_{\pmb{x}\in\Lambda_c^*}|\pmb{B}\pmb{x}|^2= (2r_{\rm pack}(\pmb{B}\pmb{G}))^2,$$
where $r_{\rm pack}(\pmb{B}\pmb{G})$ is the packing radius of the lattice $\Lambda(\pmb{B}\pmb{G})$. Let $b=b'(2r_{\rm pack}(\pmb{B}\pmb{G}))^2/m$, where $0\leq b'\leq 1$ is a constant independent of the lattice $\Lambda(\pmb{BG})$. Then, we can further upper bound (\ref{UB200}) as
 \begin{eqnarray}\label{UB300}
\Pr(E_s|\Lambda_c)&\leq&{\rm Pr}\left(\bigcup_{\pmb{x}\in\Lambda_c^*}\left\{{2(\pmb{B}'\pmb{x})^\mathsf{T}\pmb{e}}\geq{|\pmb{B}'\pmb{x}|^2}\right\}\right),
\end{eqnarray}
where
\begin{equation}\label{newB}
\pmb{B}'=\left(1-b' \right)\pmb{B}.
\end{equation}
The RHS of the upper bound (\ref{UB300}) corresponds to the probability of decoding error of a received signal $\pmb{y}=\tilde{\pmb{B}}\pmb{x}+\pmb{e}$ decoded using lattice decoding. It is clear from (\ref{newB}) that $\pmb{B}'$ is invertible. In this case, we obtain the equivalent channel output
$$\tilde{\pmb{y}}={\pmb{B}'}^{-1}\pmb{y}'=\pmb{x}+\tilde{\pmb{e}}.$$
Next, we apply the ambiguity decoder with decision region
\begin{equation}
\mathcal{E}'_{T,\gamma}\triangleq\left\{\pmb{z}\in\mathbb{R}^m:\pmb{z}^\mathsf{T}{\pmb{B}'}^\mathsf{T}\pmb{B}'\pmb{z}\leq MT(1+\gamma)\right\}.
\end{equation}
The probability of making a decoding error can then be upper bounded by
\begin{equation}\label{UB400}
\Pr(E_s|\Lambda_c)\leq {\rm Pr}(\tilde{\pmb{e}}\in\mathcal{E}'_{T,\gamma})+{\rm Pr}(\mathcal{A}(\mathcal{E}'_{T,\gamma})).
\end{equation}

Applying Lemma 1, there exists a lattice $\Lambda^\star_c$, and a translation vector $\pmb{u}_0^\star$ with error probability satisfying
\begin{equation}\label{amb_2}
\Pr(E_s,\Lambda^\star_c)\leq (1+\epsilon')2^{-T[\log\det({\pmb{B}'}^\mathsf{T}\pmb{B}')^{1/2T}-R]} +{\rm Pr}(\pmb{e}\notin \mathcal{E}'_{T,\gamma}).
\end{equation}

Now, one can show that
\begin{align}
(2r_{\rm pack}(\pmb{B}\pmb{G}))^2/m&= 2^{R_{\rm LAST}/M}{\Gamma(m/2+1)^{2/m}V(\mathcal{S}_m(2r_{\rm pack}(\pmb{B}\pmb{G})))^{2/m}\over m\pi 2^{R_{\rm LAST}/M}}\cr
&\stackrel{(a)}{=}2^{R_{\rm LAST}/M}{V(\mathcal{S}_m(2r_{\rm pack}(\pmb{B}\pmb{G})))^{2/m}\over 2V(\mathcal{R})^{2/m} \det(\pmb{B}^\mathsf{T}\pmb{B})^{2/m}}\cr
&={2^{R_{\rm LAST}/M}\over 2}{V_c^{2/m}\over V(\mathcal{R})^{2/m}}{V(\mathcal{S}_m(2r_{\rm pack}(\pmb{B}\pmb{G})))^{2/m}\over V_c^{2/m}V(\mathcal{V}(\pmb{B}))^{2/m}}\cr
&\stackrel{(b)}{\geq}{2^{[R_{\rm LAST}-R]/M}\over 2}{V(\mathcal{S}_m(2r_{\rm pack}(\pmb{B}\pmb{G})))^{2/m}\over  V(\mathcal{V}(\pmb{B}\pmb{G}))^{2/m}}\cr
&\stackrel{(c)}{=}{2^{[R_{\rm LAST}-R]/M}\over 2}{V(\mathcal{S}_m(2r_{\rm pack}(\pmb{B}\pmb{G})))^{2/m}\over  V(\mathcal{S}_m(r_{\rm eff}(\pmb{B}\pmb{G})))^{2/m}}={2^{[R_{\rm LAST}-R]/M}\over 2}{\left(2r_{\rm pack}(\pmb{B}\pmb{G})\over  r_{\rm eff}(\pmb{B}\pmb{G})\right)^2}
\end{align}
where $(a)$ follows from the fact that $V(\mathcal{R})$ is the volume of the $m$-dimensional hypersphere of radius $\sqrt{m/2}$, $(b)$ follows from the fact that for the shifted lattice code used in Lemma 1 we have
$${V_c^{2/m}\over V(\mathcal{R})^{2/m}}\geq 2^{-R/M}$$
and $(c)$ follows from the definition of the effective radius of the lattice generated using the matrix $\pmb{B}\pmb{G}$. Therefore, we can further upper bound (\ref{amb_2}) as
\begin{equation}\label{amb_3}
\Pr(E_s,\Lambda^\star_c)\leq (1+\epsilon')2^{-T[\log\det(\tilde{\pmb{B}}^\mathsf{T}\tilde{\pmb{B}})^{1/2T}-R]} +{\rm Pr}(\pmb{e}\notin \mathcal{E}'_{T,\gamma}).
\end{equation}
where $\tilde{\pmb{B}}$ is given by
\begin{equation}\label{newB2}
\tilde{\pmb{B}}=\left(1-{2b\over 2^{[R_{\rm LAST}-R]/M}(2r_{\rm pack}(\pmb{B}\pmb{G})/ r_{\rm eff}(\pmb{B}\pmb{G}))^2}\right)\pmb{B},
\end{equation}
which is valid for all values of $b< 2^{[R_{\rm LAST}-R-1]}(2r_{\rm pack}(\pmb{B}\pmb{G}_c)/r_{\rm eff}(\pmb{B}\pmb{G}_c))^2 $. Noticing that
\begin{align}
\det\left({\tilde{\pmb{B}}}^\mathsf{T}\tilde{\pmb{B}}\right)^{1/2T}&=\left(1-{2b\over 2^{[R_{\rm LAST}-R]/M}(2r_{\rm pack}(\pmb{B}\pmb{G}_c)/r_{\rm eff}(\pmb{B}\pmb{G}_c))^2}\right)^{2m/2T}\det\left({\pmb{B}}^\mathsf{T}\pmb{B}\right)^{1/2T}\cr
&=\left(1-{2b\over 2^{[R_{\rm LAST}-R]/M}(2r_{\rm pack}(\pmb{B}\pmb{G}_c)/r_{\rm eff}(\pmb{B}\pmb{G}_c))^2}\right)^{2M}\det(\pmb{I}_M+\rho(\pmb{H}^c)^\mathsf{H}\pmb{H}^c),
\end{align}
and by solving for $R$, we achieve the desired result.
\end{proof}

As discussed earlier, choosing a fixed but not very large values of $b$ may result in achieving the optimal DMT of the channel. However, lattice sequential decoders are used as an alternative to ML and lattice decoders to achieve very low decoding complexity and to do so one has to resort to large values of $b$. As will be shown in the sequel, choosing large values of $b$ may lead to a loss in the diversity gain and/or the multiplexing gain, and as a result, a loss in the optimal tradeoff.

\subsection{Achievable DMT: Variable Bias Term}
Our goal in this section is to derive the achievable outage performance for a general (not necessary fixed) bias term $b$. Denote $0\leq\lambda_1\leq\cdots\leq\lambda_M$ the eigenvalues of $(\pmb{H}^c)^\mathsf{H}{\pmb{H}^c}$. Consider $b$ as a function of $\rho$ and $\pmb{\lambda}=(\lambda_1,\cdots,\lambda_M)$, and express it as
\begin{equation}\label{bias_SNR}
b(\pmb{\lambda},\rho)={1\over 2}{\prod_{i=1}^{M}(1+\rho\lambda_i)^{1/M}\over \eta(\pmb{\lambda},\rho)^{1/M}}\left[1-\left({\eta(\pmb{\lambda},\rho)\over \prod\limits_{i=1}^{M}(1+\rho\lambda_i)}\right)^{1/2M}\right]\left(2r_{\rm pack}(\pmb{B}\pmb{G}_c)\over r_{\rm eff}(\pmb{B}\pmb{G}_c)\right)^2.
\end{equation}
Then, one can easily show that by substituting $b$ in (\ref{alpha}), we get
\begin{equation}
R_b(\pmb{\lambda},\rho)=\log\eta(\pmb{\lambda},\rho).
\end{equation}

To fully characterize the achievable outage performance of lattice sequential decoders as a function of the bias, we allow $b$ to vary with SNR and the channel eigenvalues as in (\ref{bias_SNR}). We define the outage event under lattice sequential decoding as
$$\mathcal{O}_b(\rho)\triangleq\{\pmb{H}^c: R_b(\pmb{H}^c,\rho)<R\}.$$
Denote $R=r\log\rho$. The probability that the channel is in outage, $P_{\rm out}(\rho,b)={\rm Pr}(\mathcal{O}_b(\rho))$, can be evaluated as follows:
  \begin{eqnarray}
  P_{\rm out}(\rho,b)={\rm Pr}(\log\eta(\pmb{\lambda},\rho)<R).
  \end{eqnarray}
The term $\eta (\pmb{\lambda},\rho) $ can be chosen freely between 1 (zero rate) and  $\prod_{i=1}^{M}(1+\rho\lambda_i)$ (yields the maximum achievable rate under MMSE-DFE lattice decoding). Depending on the value of $\eta(\pmb{\lambda},\rho)$ we obtain different achievable rates and hence different outage performances. However, in our analysis and for the sake of simplicity, we let
\begin{equation}\label{eta_zeta}
\eta (\pmb{\lambda},\rho) =\phi\prod_{i=1}^{M}(1+\rho\lambda_i)^{\zeta_i},
\end{equation}
where $0< \phi\leq 1$ is a constant, and $\zeta_i,\;\forall 1\leq i\leq M$, are constants that satisfy the following two constraints: $\sum_{i=1}^M \zeta_i\leq M$, and $\zeta_1\geq\zeta_2\geq\cdots\geq \zeta_M\geq 0$. For example, setting $\eta(\pmb{\lambda},\rho)=\phi\prod_{i=1}^{M}(1+\rho\lambda_i)$ (i.e., uniform values of $\zeta_i$, $\zeta_i=1$, $\forall i=1,\cdots,M$) we achieve the optimal DMT in the sense that for such choice of $\eta(\pmb{\lambda},\rho)$ we have
$$b={1\over 2}\phi^{-1/M}[1-\phi^{1/2M}],$$
which is fixed. This result agrees with Theorem 2.

Now define $\nu_i\triangleq-\log\lambda_i/\log\rho$, then
\begin{eqnarray}
    P_{\rm out}(\rho,b) &=&{\rm Pr}\left(\log\phi\prod_{i=1}^{M}\left(1+{\rho}\lambda_i\right)^{\zeta_i}<r\log\rho\right)\cr
                                &\limi{=}&{\rm Pr}\left(\sum\limits_{i=1}^{M}\zeta_i(1-\nu_i)^{+}<r\right),
                                 \end{eqnarray}
where $(x)^+=\max\{0,x\}$, and $\log\phi$ can be neglected at the high SNR regime. The typical outage event, at high SNR, can be written as
$$\mathcal{O}_b^+(\zeta_1,\cdots,\zeta_M)\triangleq\left\{\pmb{\nu}\in\mathbb{R}_+^M:\sum\limits_{i=1}^M\zeta_i(1-\nu_i)^+ < r\right\}.$$
In this case, the outage probability can be evaluated as follows:
$$P_{\rm out}(\rho,b)=\int_{\mathcal{O}_b^+(\zeta_1,\cdots,\zeta_M)}f_{\pmb{\nu}}(\pmb{\nu})\;d\pmb{\nu},$$
where $f_{\pmb{\nu}}(\pmb{\nu})$ is the joint probability density function of $\pmb{\nu}$ which, for all $\pmb{\nu}\in\mathcal{O}_b^+(\zeta_1,\cdots,\zeta_M)$, is asymptotically given by \cite{GCD}
\begin{equation}\label{pdf1}
f_{\pmb{\nu}}(\pmb{\nu})\limi{=}\exp\left(-\log(\rho)\sum\limits_{i=1}^{M}(2i-1+N-M)\nu_i\right).
\end{equation}
Applying Varadhan's lemma as in \cite{ZT}, we obtain
$$P_{\rm out}(\rho,b)\limi{=}\rho^{-d_b(r)},$$
where
$$d_b(r)=d(r,\pmb{\zeta})=\inf_{\pmb{\nu}\in\mathcal{O}_b^+(\zeta_1,\cdots,\zeta_M)}\sum\limits_{i=1}^{M}(2i-1+N-M)\nu_i.$$
where $\pmb{\zeta}=(\zeta_1,\cdots,\zeta_M)$. It is clear from the above optimization problem that $d_b(r)$ depends critically on the selected coefficients $\pmb{\zeta}$ (or equivalently $b$). Since $\zeta_i$ are ordered, one can assume without loss of generality of the optimal solution that $1\geq\nu_1\geq\cdots\geq\nu_M\geq0$. The linear optimization problem is therefore equivalent to the following problem
$$\left\{\begin{array}{ll}
\hbox{Minimize}: &\sum\limits_{i=1}^M(2i-1+N-M)\nu_i\\
\hbox{Such that}: &0\leq \nu_i\leq 1\quad\forall i\geq 2\\
&\sum\limits_{i=1}^M{\zeta}_i\nu_i\geq M-r
\end{array}
\right.$$
where $\zeta_i\in[0,M]$. We arrive now to the following results:
\begin{itemize}
\item{\textit{Case 1}:} ($0<\zeta_i< M$, and $\sum_{i=1}^M\zeta_i= M$) We have the following:
\begin{itemize}
\item If $r=0$, the optimal solution is
$$\nu_1^*=\cdots=\nu_M^*=1.$$

\item If $r\neq0$, the optimal solution is
\begin{equation}
\nu_i^*=\min\left[{1\over\zeta_i}\left(\sum\limits_{j=i}^{M}\zeta_j-r\right)^+,1\right]\quad \forall i\geq 1,
\end{equation}
\end{itemize}
and the DMT is given by
\begin{equation}
\begin{split}
d_b(0)&=MN,\\
d(r,\pmb{\zeta})&=\sum\limits_{i=1}^M(2i-1+N-M)\nu_i^*.
\end{split}
\end{equation}
An interesting remark about this DMT is that maximum diversity $d(0,\pmb{\zeta})=MN$ is independent of $\zeta_i,\forall i\geq 1$. Moreover, other than the uniform assignments of $\pmb{\zeta}=(1,\cdots,1)$, the optimal DMT \textbf{cannot} be achieved.
\\
\item{\textit{Case 2}:} ($\zeta_i= 0$ for some $i$) For such choices of $\zeta_i$, it is clear that the optimal DMT is lost, i.e., $d_b(r)<(M-r)(N-r)$ for all $r=0,1,\cdots,M$. The maximum diversity achieved in this scenario can be easily shown to be given by
$$d(0,\pmb{\zeta})=MN-\sum\limits_{i=1}^{M}(2i-1+N-M)\delta(\zeta_i),$$
where $\delta(\zeta_i)=1$ if $\zeta_i=0$ and $0$ otherwise.
\begin{ex}
Consider a LAST coded $M\times N$ MIMO channel under lattice sequential decoding with $\zeta_1=M$, and $\zeta_i=0$ for all $i\neq 1$. In this case, the achievable rate is given by $R_{b}(\rho,\lambda_1)=M\log(1+\rho\lambda_1)$. The asymptotic outage probability can be expressed as
\begin{align}
P_{\rm out}(r\log\rho)\limi{=}\Pr(M\log(1+\rho\lambda_1)<r\log\rho)&=\Pr(\rho\lambda_1<\rho^{r/M}-1)\cr
                                                                   &\limi{=}\Pr(\lambda_1<\rho^{-(1-r/M)})\cr
                                                                   &\limi{=}\rho^{-(N-M+1)(1-r/M)^+},
\end{align}
where we have used the fact that $\Pr(\lambda_1<\epsilon)=\epsilon^{N-M+1}$ as $\epsilon\rightarrow 0$ \cite{MMSE2}. Therefore, for such value of $b$, the best DMT that can be achieved by the decoder is $d_b(r)=(N-M+1)(1-r/M)^+$.
\end{ex}
\end{itemize}
Interestingly, for \textit{Case 1}, one can derive a closed form for the achievable DMT as given in the following theorem:
\begin{thm} The DMT, $d_b(r)$, for an $M$-transmit, $N$-receive antenna coded MIMO Rayleigh channel under MMSE-DFE lattice Fano/Stack sequential decoding with bias $b$ as given in (\ref{bias_SNR}) and coefficients $\zeta_i\in(0,M)$, $\forall 1\leq i\leq M$, with $\sum_i\zeta_i=M$, is the piecewise-linear function connecting the points (r(k),d(k)), $k=0,1,\cdots,M$ where
\begin{equation}\label{DMT2}
\begin{split}
r(0)&=0, \quad r(k)= \sum\limits_{i=M-k+1}^{M}\zeta_i,\;1\leq k\leq M,\\
d(k)&=(M-k)(N-k),  \quad\quad 0\leq k\leq M.
\end{split}
\end{equation}
  \end{thm}

\begin{proof}
By solving the above optimization problem, we obtain the following DMT:
\begin{equation}\label{DMT3}
d(r,\pmb{\zeta})=\begin{cases}\sum\limits_{i=1}^{M-k-1}(2i-1+N-M)+&\cr
\quad\quad\displaystyle{2(M-k)-1+N-M\over \zeta_{M-k}}\left(\sum\limits_{j=M-k}^{M}\zeta_j-r\right), & r\in[r_k,r_{k+1}],\;0\leq k\leq M-2;\cr
\displaystyle{N-M+1\over \zeta_1}\left(\sum\limits_{j=1}^{M}\zeta_j-r\right),& r\in[r_{M-1},r_M], \end{cases}
\end{equation}
where
$$r_k=\begin{cases}0, & k=0;\cr \sum\limits_{i=M-k+1}^{M}\zeta_i, & 1\leq k\leq M.\end{cases}$$
Substituting $r_k$ in (\ref{DMT3}), we get the DMT expression in (\ref{DMT2}).
\end{proof}
\begin{ex}
Consider a $2\times 2$ MIMO channel. The DMT curves achieved with respect to different values of $\zeta_i$ that correspond to \textit{Case 1} and \textit{Case 2} are illustrated in Fig.~1. Although the diversity at $r=0$ is not affected by the coefficients $\zeta_i\neq0$ ($d(0)=4$), the more unbalanced the coefficients are, the worse the DMT is.
\end{ex}

It is clear from the above analysis that by varying $\zeta_i$ and correspondingly varying $b$, one can fully control the maximum diversity and multiplexing gains achieved by such decoding scheme. Fig.~2 shows the achievable DMT curves under lattice sequential decoding for all possible values of $\zeta_i$ that satisfy the constraint $\sum_{i=1}^M\zeta_i=M$. The figures include both \textit{Case 1} and \textit{Case 2}.

Following the footsteps of \cite{GCD},  we are now ready to prove the following theorem:
\begin{thm} There exists a sequence of full-dimensional LAST codes with block length $T\geq M+N-1$ that achieves the DMT curve $d_b(r)$ under LAST coding and MMSE-DFE lattice Fano/Stack sequential decoding with variable bias term $b$ that is given in (\ref{bias_SNR}).
\end{thm}
\begin{proof}
See Appendix III.
\end{proof}

\subsection{Improving the Achievable Rate}
It is clear from (\ref{fano_rate}) that lattice sequential decoders suffer from very poor performance as $b$ becomes large (achievable rate $R_b$ could reach 0!). The question that may arise here is whether the achievable rate of the decoder can be improved especially for large values of $b$ (for which low decoding complexity is to be expected \cite{MGDC}) and hence improving the error performance.

It turns out that the way the nodes are generated in the algorithm plays an important role in improving both the achievable rate and performance of the decoder without increasing the decoding complexity. For example, Schnorr-Euchner enumeration is considered a good candidate for the use in lattice Fano/Stack sequential decoding algorithms \cite{MGDC}. If the determination of best and next best nodes in the lattice Fano/Stack sequential decoder is based on the Schnorr-Euchner search strategy, then as $b\rightarrow\infty$ the decoder reduces to the MMSE-DFE decoder \cite{MGDC}, which achieves a DMT given by $(N-M+1)(1-r/M)^+$ \cite{MMSE2}.
\begin{cor}
For a fixed non-random channel matrix $\pmb{H}^c$, the rate
\begin{equation}\label{fano_rate_2}
R_b(\pmb{H}^c,\rho)\triangleq \max\biggl\{R_{\rm LAST}(\pmb{H}^c,\rho)-2M\log\left({1+\sqrt{1+8\alpha}\over 2}\right), R_{\rm MMSE-DFE}(\pmb{H}^c,\rho)\biggr\},
 \end{equation}
is achievable by LAST coding and MMSE-DFE lattice Fano/Stack sequential decoding constructed under the Schnorr-Euchner search strategy, where $R_{\rm MMSE-DFE}(\pmb{H}^c,\rho)$ is the achievable rate of the MMSE-DFE decoder, and $\alpha$ is as defined in (\ref{alpha}).
\end{cor}

In what follows, we discuss some interesting results about low computational complexity receivers.

\subsection{MMSE-like Receivers: Large $N$ Analysis}
The main role of the bias term $b$ is to control the amount of computations performed by the decoder. The computational complexity of the lattice sequential decoder is defined as the total number of nodes visited by the decoder during the search. It has been shown in \cite{MGDC} via simulation, that there exists a value of $b$, say $b^*$, such that for all $b\geq b^*$, the computational complexity decreases monotonically with $b$. As $b\rightarrow \infty$, the number of visited nodes is always equal to $m$ (computational complexity of the MMSE-DFE decoder). In what follows, we discuss a very interesting result.

It is clear from the above analysis that increasing the bias $b$ can affect both diversity and multiplexing gains achieved by such a decoding scheme. However, we would like to show that at $r=0$ (i.e., at fixed rate $R$), there exists a lattice sequential decoding algorithm that can simultaneously achieve computational complexity $m$ and maximum diversity $d=MN$.

Consider the bias term given in (\ref{bias_SNR}) with $\eta(\pmb{\lambda},\rho)=\prod_{i=1}^M(1+\rho\lambda_i)^{\zeta_i}$ where the coefficients $0<\zeta_i<1$ are chosen according to \textit{Case 1} such that $\zeta_i=\epsilon$ for all $i$. In this case, as $\rho\rightarrow \infty$, it can be easily verified that $b\limi{=}\rho^{{(1-\epsilon)\over M}\sum_{i=1}^M(1-\alpha_i)^+}$. The probability that $b$ exceeds $\rho^{\kappa/M}$, for $0<\kappa<M$, can be evaluated as follows:
\begin{equation*}
\begin{split}
\Pr(b\geq \rho^{\kappa/M})&\limi{=}\Pr\left((1-\epsilon)\sum_{i=1}^M(1-\alpha_i)^+\geq \kappa\right)=1-\Pr\left(\sum_{i=1}^M(1-\alpha_i)^+<{\kappa\over (1-\epsilon)}\right)\cr
&\limi{=}1-\rho^{-\left(N-{\kappa\over (1-\epsilon)}\right)^+\left(M-{\kappa\over (1-\epsilon)}\right)^+}.
\end{split}
\end{equation*}
It is clearly seen that, as $N$ becomes large, with probability close to 1 the bias term $b\rightarrow \infty$ as $\rho\rightarrow \infty$. Therefore, for such choice of $\eta(\pmb{\lambda},\rho)$, at high SNR we can achieve \textit{linear} computational complexity but at the expense of losing the optimal tradeoff. However, as argued in the proof of Theorem 3, at $r=0$ we have $d=MN$. Therefore, as $\rho\rightarrow \infty$, linear computational complexity $m$ and maximum diversity gain $MN$ can be achieved simultaneously for large values of $N$. We can conclude that there exists a lattice sequential decoding algorithm that achieves ML decoder's diversity gain, $MN$, at $r=0$ (fixed rate $R$) when $N\rightarrow \infty$.

\section{Computational Complexity: Tail Distribution in the High SNR Regime}
Lattice sequential decoders are constructed as an alternative to sphere decoders (or equivalently lattice decoders) to solve the CLPS problem with much lower computational complexity. Due to the random nature of the channel matrix and the additive noise, the computational complexity of both decoders is considered difficult to analyze in general. As such, most of the work related to such analysis has been performed via first and second order statistics of the complexity \cite{HV},\cite{HV2},\cite{JB}. However, in \cite{SJ}, Seethaler \textit{et. al.} took a different path and analyzed the sphere decoder through its complexity tail distribution defined as $\Pr(C\geq L)$, where $C$ is the total number of computations performed by the decoder and $L$ is the distribution parameter. This approach follows naturally from the randomness of the computational complexity of such decoding scheme. It has been shown in \cite{SJ} that, for large $L$ (i.e., as $L\rightarrow\infty$), the complexity distribution of sphere decoder is of a Pareto-type that is given by $L^{-(N-M+1)}$.

As discussed earlier, the bias term $b$ is responsible for the performance-complexity tradeoff achieved by the lattice sequential decoders \cite{MGDC}. For example, setting $b=0$, we achieve the best performance (performance of sphere decoder) but at the expense of very large decoding complexity. On the other extreme, setting $b=\infty$, lattice sequential decoder that uses Schnorr-Euchner enumeration becomes equivalent to the MMSE-DFE decoder. Although it achieves very low decoding complexity, it suffers from poor performance. In our work, we consider the case of fixed (finite) $b$.
It turns out that for fixed but not large values of $b$, the complexity distribution's tail exponent $e(r)$ defined by
$$e(r)=\lim_{\rho\rightarrow \infty}{-\log\Pr(C\geq L)\over \log \rho},$$
is asymptotically lower bounded by the DMT achieved by the LAST coding and sequential decoding schemes, i.e., $e(r)\geq d_{\rm out}(r)$, and does not depend on the bias term at the high SNR regime. However, increasing the value of $b$ could significantly lower the computational complexity (e.g., as $b\rightarrow\infty$, $\Pr(C>L)=0$ for $L\geq m$) but at the expense of great loss in the achievable DMT.

We consider only lattice codes that are DMT optimal. Also, for the sake of simplicity we consider the Stack algorithm in analyzing the decoder's computational complexity. It must be noted that the following analysis is \textit{only} valid for finite but small values of $b$.

In this section, we would like to analyze the computational complexity of the MMSE-DFE lattice Stack sequential decoder with bias term $b>0$, particularly at the high SNR regime. We are interested in bounding the tail distribution of the decoder's computational complexity at high SNR.

\begin{thm} The asymptotic computational complexity distribution of the MMSE-DFE lattice sequential decoder in an $M\times N$ LAST coded MIMO channel with codeword length $T\geq N+M-1$, is upper bounded by the asymptotic outage probability, i.e.,
\begin{equation}
\Pr(C\geq L)\limi{\leq} \rho^{-d^*_{\rm out}(r)},
\end{equation}
for all $L$ that satisfy
\begin{equation}\label{L}
L\geq m+ \sum\limits_{k=1}^m{{(4\pi)}^{k/2} \over \Gamma(k/2+1)}{[bk+MT(1+\log\rho)]^{k/2}\over \det(\pmb{R}_{kk}^\mathsf{T}\pmb{R}_{kk})^{1/2}},
\end{equation}
where $\pmb{R}_{kk}$ is the lower $k\times k$ part of $\pmb{R}=\pmb{Q}^\mathsf{T}\pmb{BG}$, and $d^*_{\rm out}(r)=(M-r)(N-r)$.
\end{thm}

\begin{proof}
The input to the decoder, after QR preprocessing ($\pmb{B}\pmb{G}=\pmb{Q}\pmb{R}$) of (\ref{linear_model}), is given by $\pmb{y}'=\pmb{Q}^\mathsf{T}\pmb{y}=\pmb{R}\pmb{z}+\pmb{e}'$, where $\pmb{e}'=\pmb{Q}^\mathsf{T}\pmb{e}$. Let $\mu_{\min}=\min\{0,b-|{\pmb{e}'}_1^1|^2,2b-|{\pmb{e}'}_1^2|^2,\ldots,bm-|{\pmb{e}'}_1^m|^2\}$ be the minimum metric that corresponds to the transmitted path. Without loss of generality, we assume that $N\geq M$. Due to lattice symmetry, we assume that the all zero codeword, i.e., $\pmb{0}$, was transmitted.

First, let
$$C=\sum_{k=1}^{m}\sum_{\pmb{z}_1^k\in\mathbb{Z}^k}\phi(\pmb{z}_1^k),$$
be a random variable that denotes the total number of visited nodes during the search, where $\phi(\pmb{z}_1^k)$ is the indicator function defined by
$$\phi(\pmb{z}_1^k)=\begin{cases} 1, &\text{if node $\pmb{z}_1^k$ is extended;}\cr
                                                          0, &\text{otherwise.}\end{cases}$$
                                                          In this case, the computational complexity tail distribution can be expressed as $\Pr(C\geq L)$, where $L$ is the distribution parameter.
                                                          Now, a node at level $k$, i.e., $\pmb{z}_1^k$, may be extended by the Stack decoder if $\mu(\pmb{z}_1^k)>\mu_{\min}$, or equivalently, if $|{\pmb{e}'}_1^k-\pmb{R}_{kk}\pmb{z}_1^k|^2\leq bk-\mu_{\min}$. The difficulty in analyzing the computational complexity of the lattice Stack sequential decoder stems from the fact that the distribution of the partial matrix $\pmb{R}_{kk}$ is hard to obtain in general. Another factor that may complicate the analysis is $\mu_{\min}$ which is a noise dependent term. However, we can simplify the analysis by considering the following.

First, the complexity tail distribution can be upper bounded as
\begin{equation}\label{UBb1}
\Pr(C\geq L)\leq \Pr(C\geq L,|\pmb{e}'|^2\leq R_s^2)+\Pr(|\pmb{e}'|^2>R_s^2).
\end{equation}
where $R_s^2>0$.

Next, we would like to further upper bound the second term in the RHS of (\ref{UBb1}). We can first write $\phi(\pmb{z}_1^k)$ as
$$\phi(\pmb{z}_1^k)=\begin{cases} 1, &\text{if $|{\pmb{e}'}_1^k-\pmb{R}_{kk}\pmb{z}_1^k|^2\leq bk-\mu_{\min}$;}\cr
                                                          0, &\text{otherwise,}\end{cases}$$

Given $|\pmb{e}'|^2\leq R_s^2$, and by noticing that $-(\mu_{\min}+|\pmb{e}'|^2)\leq 0$, we obtain
\begin{equation}
\sum_{\pmb{z}_1^k\in\mathbb{Z}^k}\phi(\pmb{z}_1^k)\leq\sum_{\pmb{z}_1^k\in\mathbb{Z}^k}\phi^{'}(\pmb{z}_1^k),
\label{aa}
\end{equation}
where
\begin{equation}\label{C111}
\phi^{'}(\pmb{z}_1^k)=\begin{cases} 1, &\text{if $|{\pmb{e}'}_1^k-\pmb{R}_{kk}\pmb{z}_1^k|^2\leq bk+R_s^2$;}\cr
                                                          0, &\text{otherwise.}\end{cases}
                                                          \end{equation}
Now, let
$$
\phi^{''}_k(\pmb{z})=\begin{cases} S_k, &\text{if $|\pmb{e}'-\pmb{R}\pmb{z}|^2\leq bm-\mu_{\min}$;}\cr
                                                          0, &\text{otherwise,}\end{cases}
$$
where
\begin{equation}\label{Sk}
S_k=\sum_{\pmb{z}_1^k\in\mathbb{Z}^k}\phi^{'}(\pmb{z}_1^k),
\end{equation}
then it can be easily shown that
$$C\leq\sum\limits_{k=1}^m\sum_{\pmb{z}\in\mathbb{Z}^m}\phi^{''}_k(\pmb{z})\leq\sum\limits_{k=1}^m\sum_{\pmb{x}\in\Lambda_c}\tilde{\phi}_k(\pmb{x}),$$
where
$$\tilde{\phi}_k(\pmb{x})=\begin{cases} S_k, &\text{if $|\pmb{B}\pmb{x}|^2-2(\pmb{B}\pmb{x})^\mathsf{T}\pmb{e}\leq bm$;}\cr
                                                          0, &\text{otherwise,}\end{cases}
                                                        . $$
Notice the independence of the above upper bound on $\mu_{\min}$. Consider now the following lemma:
\begin{lem}
In the lattice Stack sequential decoder with finite bias $b>0$, the number of visited nodes at level $k$, given that $|\pmb{e}'|^2\leq MT(1+\log\rho)$, can be upper bounded by
\begin{equation}\label{lemma}
\sum_{\pmb{z}_1^k\in\mathbb{Z}^k}\phi(\pmb{z}_1^k)\leq S_k\leq {{(4\pi)}^{k/2}\over \Gamma(k/2+1)}{[bk+MT(1+\log\rho)]^{k/2}\over \det(\pmb{R}_{kk}^\mathsf{T}\pmb{R}_{kk})^{1/2}},
\end{equation}
where $S_k$ is as defined in (\ref{Sk}).
\end{lem}
\begin{proof}
See Appendix III.
\end{proof}
%Now, one may approximate $S_k$ as
%$$S_k\approx {V(\mathcal{S}_{\pmb{0}}^k(\sqrt{7}R_s))\over \det(\pmb{R}_{kk}^\mathsf{T}\pmb{R}_{kk})^{1/2}}=C'_k.$$

For a given lattice $\Lambda_c$, we have
\begin{equation}
\begin{split}
{\rm Pr}(C\geq L|\Lambda_c,|\pmb{e}'|^2\leq MT(1+\log\rho))&\leq\Pr(\tilde{C}\geq L-m|\Lambda_c,|\pmb{e}'|^2\leq MT(1+\log\rho))\cr
&\leq{\mathsf{E}_{\pmb{e}'}\{\tilde{C}|\Lambda_c,|\pmb{e}'|^2\leq MT(1+\log\rho)\}\over L-m}, \quad \text{for $L>m$,}
\end{split}
\end{equation}
where the last inequality follows from using Markov inequality, and $\tilde{C}$ is defined as
$$\tilde{C}=\sum_{k=1}^{m}\sum_{\pmb{z}_1^k\in\mathbb{Z}^k\backslash\{\pmb{0}\}}\phi(\pmb{z}_1^k),$$
since we have assumed that the all-zero lattice point was transmitted.

The conditional average of $\tilde{C}$ with respect to the noise can be further upper bounded as
\begin{equation}\label{NUB1}
\begin{split}
\mathsf{E}_{\pmb{e}'}\{\tilde{C}|\Lambda_c,|\pmb{e}'|^2\leq MT(1+\log\rho)\}&\leq \sum\limits_{k=1}^m S_k\sum\limits_{\pmb{x}\in\Lambda_c^*}\Pr(|\pmb{B}\pmb{x}|^2-2(\pmb{B}\pmb{x})^{\mathsf{T}}\pmb{e}<bm)
\end{split}
\end{equation}
Therefore, we have
\begin{equation}\label{EE1}
\Pr(C\geq L|\Lambda_c,|\pmb{e}|^2\leq MT(1+\log\rho))\leq{\sum_{k=1}^m S_k\over L-m}\sum\limits_{\pmb{x}\in\Lambda_c^*}\Pr(|\pmb{B}\pmb{x}|^2-2(\pmb{B}\pmb{x})^{\mathsf{T}}\pmb{e}<bm).
\end{equation}

Following the proof of Theorem 2 (see Appendix I), and by averaging over the ensemble of random lattices we get, for $L>m+\sum_{k=1}^m S_k$
\begin{equation}\label{UBa1}
\Pr(C\geq L)\limi{\leq} \rho^{-T[\sum_{j=1}^{\min\{M,N\}}(1-\alpha_j)^+-{r}]}.
\end{equation}
Define $\mathcal{B}=\{\pmb{\nu}\in\mathbb{R}_+^M:\nu_1\geq\cdots\geq\nu_M\geq 0,\; \sum_{i=1}^{M}(1-\nu_i)^{+}<r\}$. By separating the event $\{\pmb{\nu}\in \mathcal{A}\}$ from its complement, we obtain:
\begin{equation}\label{UU1}
\Pr(C\geq L)\leq  {\rm Pr}(\pmb{\nu}\in\mathcal{A})+\Pr(|\pmb{e}'|^2>MT(1+\log\rho))+
\Pr(C\geq L,\pmb{\nu}\in\overline{\mathcal{A}},|\pmb{e}'|^2\leq MT(1+\log\rho))
\end{equation}

The behavior of the first term in (\ref{UU1}) at high SNR is $\rho^{-d^*_{\rm out}(r)}$, where $d^*_{\rm out}(r)$ is as defined in Theorem 1. The second term can be shown to be upper bounded by $\rho^{-d^*_{\rm out}(r)}$ (see \cite{GCD}). Averaging the third term over the channels in $\overline{\mathcal{A}}$ set, we obtain,
\begin{equation}\label{Nm1}
\Pr(C\geq L)\limi{\leq}\rho^{-d^*_{\rm out}(r)}+\int_{\overline{\mathcal{A}}}f_{\pmb{\nu}}(\pmb{\nu})\Pr(C\geq L|\pmb{\nu})\;d\pmb{\nu}\limi{\leq}\rho^{-d^*_{\rm out}(r)},
\end{equation}
for all $L\geq m+\sum_{k=1}^mS_k$, where $f_{\pmb{\nu}}(\pmb{\nu})$ is the joint probability density function of $\pmb{\nu}$ defined in (\ref{pdf1}).
\end{proof}

The above results reveal that if the number of computations performed by the decoder exceeds
\begin{equation}\label{L0}
L_0= m+ \sum\limits_{k=1}^m{{(7\pi)}^{k/2} \over \Gamma(k/2+1)}{[bk+MT(1+\log\rho)]^{k/2}\over \det(\pmb{R}_{kk}^\mathsf{T}\pmb{R}_{kk})^{1/2}},
\end{equation}
the complexity distribution of the lattice sequential decoder at high SNR is upper bounded by the asymptotic outage probability. Now, if a ``time-out'' limit is imposed at the decoder to terminate the search when the number of computations exceeds this limit, then $L_0$ represents the minimum value that should be set by the decoder without resulting in a loss in the optimal DMT. To see this, suppose that the lattice (Stack) sequential decoder imposes a time-out limit so that the search is terminated once the number of computations reaches $L_0$, and hence the decoder declares an error. Assuming $E'_s$ is the event that the decoder performs an error when $C<L_0$, in this case, the average error probability is given by
\begin{equation}\label{Pe_C}
P_e(\rho)=\Pr(E'_s \cup \{C\geq L_0\})\leq \Pr(E'_s) + \Pr(C\geq L_0)\limi{\leq} \rho^{-d^*_{\rm out}(r)}.
\end{equation}
This can be very beneficial in two-ways MIMO communication systems (e.g, MIMO automatic repeat request~\cite{MIMOARQ}), where the feedback channel can be used to eliminate the decoding failure probability. In applications where there is a hard-limit on the buffer size, the decoder declares an error when the complexity goes above the limit.

It should be noted that the above analysis does not yield the full picture of the decoder's complexity in general. As mentioned previously, the complexity of the decoder depends critically on the bias $b$ chosen in the algorithm. Unfortunately, it is still unclear how the SNR exponent $e(r)$ is affected by the value $b$ in general. However, as $b\rightarrow \infty$, the MMSE-DFE lattice sequential decoder under Schnorr-Euchner enumeration becomes equivalent to the MMSE-DFE decoder \cite{DGC}. The total number of computations performed by this decoder is always equal to $m$. This corresponds to an SNR exponent $e(r)=\infty$. Thus, we can conclude that, at high SNR, as $b$ increases the SNR exponent $e(r)$ increases as well.

Another criterion that is used to characterize the computational complexity of such a decoder is through its average complexity. Since $L_0$ is random, it would be interesting to calculate the minimum average number of computations required by the decoder to terminate the search. This is considered next.

\section{Average Computational Complexity}
It is to be expected that when the channel is ill-conditioned (i.e., in outage) the computational complexity becomes extremely large. Moreover, when the channel is in outage it is highly likely that the decoder performs an erroneous  detection. However, when the channel is \textit{not} in outage, there is still a non-zero probability that the number of computations will become large (see (\ref{Nm1}) and (\ref{L0})). As such, it is sometimes desirable to terminate the search even when the channel is not in outage. Therefore, we would like to determine the \textit{minimum} average number of computations that is required in order for the decoder to determine when to terminate the search.

In other words, we would like to find the minimum average number of computations that is required by the decoder to achieve the optimal DMT. This can be expressed as
\begin{eqnarray}\label{d2}
L_{\rm out}=\mathsf{E}\{L_0(\pmb{H}^c\in\overline{\mathcal{O}})\}.
\end{eqnarray}

Before we do that, we would like first to study the asymptotic behavior of $L_0$. As mentioned in Section I, we focus our analysis on nested LAST codes, specifically LAST codes that are generated using construction A which is described below (see \cite{Loe}).

We consider the Loeliger ensemble of mod-$p$ lattices, where $p$ is a prime. First, we generate the set of all lattices given by
$$\Lambda_p=\kappa (\mathsf{C}+p\mathbb{Z}^{2MT})$$
where $p\rightarrow \infty$, $\kappa\rightarrow 0$ is a scaling coefficient chosen such that the fundamental volume $V_f=\kappa^{2MT}p^{2MT-1}=1$, $\mathbb{Z}_p$ denotes the field of mod-$p$ integers, and $\mathsf{C}\subset\mathbb{Z}_p^{2MT}$ is a linear code over $\mathbb{Z}_p$ with generator matrix in systematic form $[\pmb{I}\;\pmb{P}^\mathsf{T}]^\mathsf{T}$. We use a pair of self-similar lattices for nesting. We take the shaping lattice to be $\Lambda_s=\phi\Lambda_p$, where $\phi$ is chosen such that the covering radius is $1/2$ in order to satisfy the input power constraint. Finally, the coding lattice is obtained as $\Lambda_c=\rho^{-r/2M}\Lambda_s$. Interestingly, one can construct a generator matrix of $\Lambda_p$ as (see \cite{Conway})
\begin{equation}
\pmb{G}_p=\kappa\begin{pmatrix}
\pmb{I} & \pmb{0}\\
\pmb{P} & p\pmb{I}
 \end{pmatrix},
\end{equation}
which has a lower triangular form. In this case, one can express the generator matrix of $\Lambda_c$ as $\pmb{G}=\rho^{-r/2M}\pmb{G}'$, where $\pmb{G}'=\zeta\pmb{G}_p$. Thanks to the lower triangular format of $\pmb{G}$. If $\pmb{M}$ is an $m\times m$ arbitrary full-rank matrix, and $\pmb{G}$ is an $m\times m$ lower triangular matrix, then one can easily show that
\begin{equation}\label{matrix_kk}
\det[(\pmb{MG})_{kk}]= \det(\pmb{M}_{kk})\det(\pmb{G}_{kk}),
\end{equation}
where $(\pmb{MG})_{kk}$, $\pmb{M}_{kk}$, and $\pmb{G}_{kk}$, are the lower $k\times k$ part of $\pmb{MG}$, $\pmb{M}$, and $\pmb{G}$, respectively.

Using the above result, one can express the determinant that appears in (\ref{L0}) as
\begin{equation}\label{matrix_kk2}
\det(\pmb{R}_{kk}^\mathsf{T}\pmb{R}_{kk})=\det(\pmb{B}_{kk}^\mathsf{T}\pmb{B}_{kk})\det(\pmb{G}_{kk}^\mathsf{T}\pmb{G}_{kk})=\rho^{-rk/2M}\det(\pmb{B}_{kk}^\mathsf{T}\pmb{B}_{kk})\det({\pmb{G}'}_{kk}^\mathsf{T}{\pmb{G}'}_{kk})
\end{equation}
Let $\mu_1\leq \mu_2\leq\cdots\leq\mu_k$ be the ordered nonzero eigenvalues of $\pmb{B}_{kk}^\mathsf{T}\pmb{B}_{kk}$, for $k=1,\cdots,m$. Then,
$$\det(\pmb{B}_{kk}^\mathsf{T}\pmb{B}_{kk})=\prod\limits_{j=1}^k\mu_j.$$
Note that for the special case when $k=m$ we have $\mu_{2(j-1)T+1}=\cdots=\mu_{2jT}=1+\rho\lambda_j((\pmb{H}^c)^\mathsf{H}\pmb{H}^c)$, for all $j=1,\cdots,M$.

Denote $\alpha'_i=-\log\mu_i/\log\rho$. Using (\ref{matrix_kk2}), one can asymptotically express $L_0$ as
\begin{equation}
L_0=m+(\log\rho)^{m/2}\sum\limits_{k=1}^m(\log\rho)^{k/2}\rho^{c_k},
\end{equation}
where
\begin{equation}
c_k={1\over 2}\sum\limits_{j=1}^k \left({r\over M} - \alpha'_j\right)^+.
\end{equation}
Now, since $c_k$ is non-decreasing in $k$, we have at high SNR
\begin{equation}
L_0=m+(\log\rho)^{m/2}\rho^{c_m},
\end{equation}
where
$$c_m=T\sum\limits_{i=1}^M \left(\displaystyle{r\over M}-(1-\alpha_i)^+\right)^+.$$

The average of $L_0$ at high SNR (averaged over the channel statistics) when the channel is not in outage is given by
\begin{align*}
\mathsf{E}\{L_0(\pmb{H}^c\in\overline{\mathcal{O}})\}&=\int\limits_{\pmb{\alpha}\in\overline{\mathcal{O}}}L_0 f_{\pmb{\alpha}}(\pmb{\alpha})\;d\pmb{\alpha}\\
&=m+(\log\rho)^{m/2}\int\limits_{\pmb{\alpha}\in\overline{\mathcal{O}}}\exp\biggl(\log\rho\biggl[T\sum\limits_{i=1}^M\biggl({r\over M}-(1-\alpha_i)^+\biggr)^+-\\
&\qquad \qquad \qquad\qquad\sum\limits_{i=1}^M (2i-1+N-M)\alpha_i\biggr]\biggr)\;d\pmb{\alpha}\\
&=m+(\log\rho)^{m/2}\rho^{l(r)},
\end{align*}
where $\overline{\mathcal{O}}=\left\{\pmb{\alpha}\in\mathbb{R}_+^M: \sum_{i=1}^{M}(1-\alpha_i)^{+}\geq r\right\}$, and
\begin{equation}\label{l'_r}
l(r)=\max_{\pmb{\alpha}\in\overline{\mathcal{O}}} \left[T\sum\limits_{i=1}^M\left({r\over M}-(1-\alpha_i)^+\right)^+-\sum\limits_{i=1}^M (2i-1+N-M)\alpha_i\right].
\end{equation}
It is not so difficult to see that the optimal channel coefficients that maximize (\ref{l'_r}) are
$$\alpha_i^*=1, \quad \hbox{for }i=1,\cdots,M-k,$$
and
$$\alpha_i^*=0, \quad \hbox{for }i=M-k+1,\cdots,M,$$
i.e., the same $\pmb{\alpha}^*$ that achieves the optimal DMT of the channel.
%For the case that $r$ is not an integer, say $r\in(k,k+1)$, we have
%$$\alpha_i^*=1, \quad \hbox{for }i=1,\cdots,M-k-1,$$
%$$\alpha_i^*=0, \quad \hbox{for }i=M-k+1,\cdots,M.$$
%and
%$$\alpha^*_i=k+1-r.$$
Substituting $\pmb{\alpha}^*$ in (\ref{l'_r}), we get
\begin{equation}\label{exponent_MMSE}
l(r)={Tr(M-r)\over M}-(M-r)(N-r),
\end{equation}
for $r=0,1,\cdots,M$. In this case, the asymptotic minimum average computational complexity that is required by the decoder to achieve near-optimal performance (as well as the optimal DMT) can be expressed as
\begin{equation}\label{UB_MMSE}
L_{\rm out}=2MT+(\log\rho)^{MT}\rho^{l(r)}.
\end{equation}
%\footnote{According to the bound (\ref{EE1}), one can see that for all values of $L>L_0$, the RHS of (\ref{EE1}) becomes equivalent to the upper bound on the performance of the lattice decoder. As such, by terminating the search when the number of computations reaches $L_0$ and declaring a decoding error, the lattice sequential decoder}

The above interesting result indicates that if the ``average'' number of computations performed by the decoder exceeds $L_{\rm out}$, the decoder can terminate the search without affecting the optimal DMT. We discuss here some special cases about the behavior of $L_{\rm out}$ in terms of the system parameters: $\rho$, $M$, $N$, and $r$. Consider the case of $M=N$. Assuming the use of an optimal random nested LAST code of codeword length $T$ and fixed rate $R$, i.e., $r=0$. In this case, one can see that $l_{\rm mmse}(0)<0$ irrespective to the value of $T$, i.e., the average complexity is bounded for all $T$. It is clear that the term $(\log\rho)^{2MT}\rho^{-NM}$ decays quickly to $0$ as $\rho\rightarrow\infty$. The simulation results (introduced next) agree with the above analysis.

It is interesting to note that, there exists a \textit{cut-off} multiplexing gain, say $r_0$, such that the average computational complexity of the decoder remains bounded as long as we operate below such value. This value can be easily found by setting $l_{\rm mmse}(r_0)=0$. This results in
$$r_0=\biggl\lfloor{MN\over M+T}\biggr\rfloor.$$

However, it should be noted that the above cut-off multiplexing gain corresponds to a sequential decoding algorithm that uses a fixed bias. If we need to operate at $r>r_0$, larger values of bias term must be used. In fact, one must let $b$ to scale with SNR as $b=\rho^\epsilon$ for some $\epsilon>0$ in order to keep the average complexity bounded when operating beyond $r_0$. However, according to the analysis provided in Section III, this causes a loss in the optimal tradeoff. Therefore, the lattice sequential decoder provides a systematic approach for tradeoff DMT, cut-off multiplexing gain, and complexity.

Another way to reduce the computational complexity without the need of increasing the bias value, is to increase the number of receive antennas $N$. If we let $N\rightarrow \infty$, then one can achieve a multiplexing gain $r_0=M$ which is the maximum multiplexing gain achieved by the channel.

To see the great advantage of using the lattice sequential decoder with constant bias term over the lattice decoder implemented via sphere decoding algorithms, we compare the average computational complexity of both decoders when MMSE-DFE is presented. It has been shown in \cite{WD2} that, for moderate-to-high SNR, the average computations performed by the MMSE-DFE sphere decoder when the channel is not in outage, say $L_{\rm sphere}$ for a system with $m=2MT$ signal dimension is given by (assuming fixed rate $r=0$)
\begin{equation}
L_{\rm sphere}=2MT+ {(\log\rho)^{2MT}\over \rho^{MN}}.
\end{equation}
The ratio of the asymptotic average complexity of both decoders, say $\gamma$, is given by
$$\gamma={L_{\rm sphere}\over L^{\rm mmse}_{\rm sequential}}={2MT+(\log\rho)^{2MT}/ \rho^{MN}\over 2MT+(\log\rho)^{MT}/\rho^{MN}}.$$
This is a huge saving in computational complexity, especially for large signal dimensions  and moderate-to-high SNR. For example, consider the case of a $3\times 3$ LAST coded MIMO system with $T=5$.  At $\rho=10^3$ (30 dB), we have $\gamma\approx 31$, i.e., the sphere decoder's complexity is about 31 times larger than the complexity of the lattice sequential decoder. As will be shown in the sequel, simulation results agree with the above theoretical results. For $\rho < 30$ dB, one would expect the ratio $\gamma\gg 31$. For extremely high SNR values (e.g., $\rho\gg 30$ dB), it seems that $\gamma\rightarrow 1$ as $\rho\rightarrow\infty$.

\section{Numerical Results}
Throughout the simulation study, the fading coefficients are generated as independent identically distributed circularly symmetric complex Gaussian random variables. The LAST code is obtained as an $(m=2MT,p,k)$ Loeliger construction (refer to \cite{Loe} for a detailed description of the linear code obtained via Construction A).

In Fig.~3, we compare the performance in terms of the frame error rate of a MIMO system with $M=N=2$, $T=3$ and rate $R=4$ bits per channel use (bpcu) under naive and MMSE-DFE lattice sequential decoding. For both decoders we fix the bias term to $b=0.6$. It is clear that the MMSE-DFE lattice sequential decoder outperforms the naive one, where the former achieves diversity order of 4 (the maximum diversity gain achieved by the channel) and the latter achieves diversity order of 2. To validate the achievability of the optimal DMT with LAST coding and MMSE-DFE lattice sequential decoding, we consider the performance of a MIMO system with $M=N=2$, $T=3$ for different rates $R=4,\;8,\;10.34$ bpcu, which is illustrated in Fig.~4. The constant gap between the outage probability and the error performance for different $R$ confirms our theoretical results.

Fig.~5 and Fig.~6 show the effect of increasing the bias term on diversity order and average computational complexity (number of visited nodes during the search) achieved by lattice sequential decoding. As discussed earlier, increasing the bias term in the decoding algorithm significantly reduces decoding complexity but at the expense of losing diversity. For the $2\times2$ LAST coded MIMO system with $T=3$, as $b\rightarrow\infty$ we achieve linear computational complexity $m=12$ for all SNR, and diversity order 1. For sequential decoding algorithms that implement the Schnorr-Euchner enumeration, this corresponds to the performance and the complexity of the MMSE-DFE decoder.

In our computational complexity distribution simulation, we consider a MIMO system with $M=N=2$, $T=3$ for different rates $R=4,\;8$ bits per channel use. First, the frame error rate of the MMSE-DFE lattice sequential decoder is plotted in Fig.~7.$(a)$ when $b=0.6$. The computational complexity distribution $\Pr(C>L)$ is plotted for such a decoder at different rates when $L$ is allowed to scale with the SNR as $L=\rho$ (see Fig.~7.$(b)$). It is clear from both figures that the curves which correspond to the error probability and the computational complexity distribution match in slope, i.e., they both exhibit the same behavior at high SNR. Equivalently, both curves have the same SNR exponent. This basically agrees with the derived theoretical results.

The complexity saving advantage that lattice sequential decoders posses over lattice (sphere) decoders is depicted in Fig.~8 and Fig.~9, for the same LAST coded MIMO channel with $R=4$ bits per channel use. One can notice the amount of computations saved by lattice sequential decoders for all values of SNR, especially for large signal dimensions (see Fig.~10). Even at high SNR, the sphere decoder still exhibits large decoding complexity compared to the lattice sequential decoder. For example, as depicted in Fig.~10, at $\rho=30$ dB, the average complexity of the sphere decoder is about 30 times the complexity of the lattice sequential decoder for an optimal LAST coded MIMO system with dimension $m=30$. This is achieved at the expense of small loss in performance ($\sim$0.6 dB). This agrees with the derived theoretical results.

Fig.~10 shows how the average complexity of the MMSE-DFE lattice sequential decoder decays with the SNR irrespective to the codeword length $T$ for a fixed rate $R$, i.e., for $r=0$. This agrees with the theoretical results derived in the previous section. Finally, Fig.~11 proves (by simulation) the fact that the MMSE-DFE lattice sequential decoder has a cut-off rate such that the average complexity of the decoder remains bounded as long as we operate below it. The figure shows that for fixed $M$, $N$, and $T$, if we increase the rate, the average complexity increases as well and becomes unbounded even at high SNR.

\section{Summary}
In this chapter, we have provided a complete analysis for the performance limits of the lattice Fano/Stack sequential decoder applied to the LAST coded MIMO system. The achievable rate of  the channel is derived. It turns out that the achievable rate under lattice sequential decoding depends critically on the decoding parameter, the bias term. The bias term is responsible for the excellent performance-complexity tradeoff achieved by such decoding scheme. For fixed values of the bias, it has been shown that the optimal tradeoff of the channel can be achieved. As the bias grows without bound, lattice sequential decoding achieves linear computational complexity, where the total number of visited nodes during the search is always equal to the lattice code dimension. As such, lattice sequential decoders bridge the gap between lattice (sphere) decodes and low complexity receivers (e.g., the MMSE-DFE decoder). At high SNR, it was argued that there exists a lattice sequential decoding algorithm that can achieve maximum diversity gain at very low multiplexing gain, especially for large number of receive antennas.

We have also provided a complete analysis for the computational complexity of the lattice sequential decoder applied to the LAST coded MIMO systems at the high SNR regime. It has been shown that for the MMSE-DFE lattice sequential decoder, if the number of computations performed by the decoder exceeds a certain limit, then the complexity's tail distribution becomes dominated by the outage probability with an SNR exponent that is equivalent to the DMT achieved by the corresponding coding and decoding schemes. The tradeoff of the channel is naturally extended to include decoding complexity. Moreover, the decoder asymptotic average computational complexity has also been analyzed. Finally, it has been shown that there exists a cut-off multiplexing gain for which the average complexity remains bounded as long as we operate below such value.

\section*{Appendix I\\Proof of Theorem 2}
The input to the decoder, after QR preprocessing ($\pmb{B}\pmb{G}=\pmb{Q}\pmb{R}$) of (\ref{linear_model}), is given by $\pmb{y}'=\pmb{Q}^\mathsf{T}\pmb{y}=\pmb{R}\pmb{z}+\pmb{e}'$, where $\pmb{e}'=\pmb{Q}^\mathsf{T}\pmb{e}$. Let $E_s$ be the event that the lattice Stack sequential decoder makes an erroneous detection. Due to lattice symmetry, we assume that the all zero codeword $\pmb{0}$ was transmitted. Now, any sequence $\pmb{x}=\pmb{Gz}\neq\pmb{0}$, $\pmb{x}\in\Lambda_c$ can be decoded as the closest lattice point by the decoder only if its metric $\mu(\pmb{z}_1^m)$ is greater than $\mu_{\min}$. Therefore, for a given lattice $\Lambda_c$,
\begin{equation}\label{UB1}
\begin{split}
\Pr(E_s|\Lambda_c)&\leq\sum\limits_{\pmb{z}\in\mathbb{Z}^m\backslash\{\pmb{0}\}}\Pr(\mu(\pmb{z}_1^m)>\mu_{\min})\\
&=\sum\limits_{\pmb{z}\in\mathbb{Z}^m\backslash\{\pmb{0}\}}\Pr(|{\pmb{e}'}-\pmb{R}\pmb{z}|^2<bm-\mu_{\min}).
\end{split}
\end{equation}
where $\mu_{\min}=\min\{0,b-|{\pmb{e}'}_1^1|^2,2b-|{\pmb{e}'}_1^2|^2,\ldots,bm-|{\pmb{e}'}_1^m|^2\}$ is the minimum metric that corresponds to the transmitted path. The upper bound in (\ref{UB1}) follows from the union bound, and due to the fact that in general, $\mu(\pmb{z}_1^m)>\mu_{\min}$ is just a necessary condition for $\pmb{x}$ to be decoded by the lattice Stack sequential decoder. By noticing that $-(\mu_{\min}+|{\pmb{e}'}|^2)\leq 0$, we get
\begin{equation}\label{UB3}
\Pr(E_s|\Lambda_c)\leq\sum\limits_{\pmb{x}\in\Lambda_c^{*}}\Pr(|\pmb{B}\pmb{x}|^2-2(\pmb{B}\pmb{x})^{\mathsf{T}}\pmb{e}'<bm),
\end{equation}
where $\Lambda_c^*=\Lambda_c\backslash\{\pmb{0}\}$. Note the independence of the upper bound (\ref{UB3}) on $\mu_{\min}$. We would like now to upper bound the term inside the summation in (\ref{UB3}). The difficulty here stems from the non-Gaussianity of the random vector ${\pmb{e}}'$ for any finite $T$. However, one can show (see \cite{EZ2} and \cite{GCD}) that for a well-constructed lattice the probability density function of the noise vector $\pmb{e}$, $f_{\pmb{e}}(\pmb{\nu})\leq \beta_m f_{\tilde{\pmb{e}}}(\pmb{\nu})$, where $\tilde{\pmb{e}}\sim\mathcal{N}(\pmb{0},0.5\pmb{I})$, and $\beta_m$ is a constant (has no effect at high SNR). Following the footsteps of \cite{GCD}, it can be shown that by appropriately constructing a nested LAST code we have that
\begin{equation}\label{UB4}
\Pr(E_s|\Lambda_c)\leq \beta_m\sum\limits_{\pmb{x}\in\Lambda_c^{*}}\Pr(|\pmb{B}\pmb{x}|^2-2(\pmb{B}\pmb{x})^{\mathsf{T}}\tilde{\pmb{e}}<bm),
\end{equation}
where $\tilde{\pmb{e}}\sim\mathcal{N}(0,0.5\pmb{I}_m)$, and $\beta_m$ is a constant independent of $\rho$. Using Chernoff bound,
\begin{equation}\label{UB5}
\Pr(|\pmb{B}\pmb{x}|^2-2(\pmb{B}\pmb{x})^{\mathsf{T}}\tilde{\pmb{e}}<bm)\leq\begin{cases}e^{-|\pmb{B}\pmb{x}|^2/8}e^{bm/4}, & \text{$|\pmb{B}\pmb{x}|^2>bm$;}\cr 1, &\text{$|\pmb{B}\pmb{x}|^2\leq bm.$} \end{cases}
\end{equation}
By taking the expectation over the ensemble of random lattices (see \cite{Loe}, Theorem 4),
\begin{equation}\label{Mink1}
\begin{split}
{\Pr}(E_s)=\mathsf{E}_{\Lambda_c}\{\Pr(E_s|\Lambda_c)\}&\leq{\beta_m\over V_c}\Bigg\{\int\limits_{|\pmb{B}\pmb{x}|^2<bm} d\pmb{x}+e^{bm/4}\int\limits_{|\pmb{B}\pmb{x}|^2>bm}e^{-|\pmb{B}\pmb{x}|^2/8}\;d\pmb{x}\Bigg\}\\
&\leq{\beta_m\over V_c}\Bigg\{{\pi^{m/2}(bm)^{m/2}\over\Gamma(m/2+1) \det(\pmb{B}^{\mathsf{T}}\pmb{B})^{1/2}}+ {(8\pi)^{m/2}e^{bm/4}\over  \det(\pmb{B}^{\mathsf{T}}\pmb{B})^{1/2}}\Bigg\}.
\end{split}
\end{equation}
Next, we make use of the fact that for nest lattice codes we have that (see \cite{EZ2})
$$|\mathcal{C}(\Lambda_c,\mathcal{R})|=2^{RT}={V(\mathcal{R})\over V_c}.$$
Also, it is easy to verify that
$$\det(\pmb{B}^{\mathsf{T}}\pmb{B})=\left(\det\left(\pmb{I}+{\rho\over M}(\pmb{H}^c)^\mathsf{H}\pmb{H}^c\right)\right)^{2T}.$$
Denote $R=r\log\rho$ and $0\leq\lambda_1\leq\cdots\leq\lambda_{\min\{M,N\}}$ the eigenvalues of $(\pmb{H}^c)^\mathsf{H}{\pmb{H}^c}$, then, the bound (\ref{Mink1}) can be rewritten as (conditioned on channel statistics)
\begin{equation}\label{10}
{\Pr}(E_s|\pmb{\nu})\limi{\leq} \mathcal{K}(m,b)\rho^{-T[\sum_{j=1}^{\min\{M,N\}}(1-\nu_j)^{+}-{r}]},
\end{equation}
where $\pmb{\nu}=(\nu_1,\cdots,\nu_{\min\{M,N\}})$, $\nu_i\triangleq-\log\lambda_i/\log\rho$, $(x)^+=\max\{0,x\}$, and $\mathcal{K}(m,b)$ is a constant independent of $\rho$. Now, define the set
\begin{equation} \label{outage2}
\mathcal{B}=\left\{\pmb{\nu}\in\mathbb{R}_+^{\min\{M,N\}}:\nu_1\geq\cdots\geq\nu_{\min\{M,N\}}\geq 0,\; \sum_{i=1}^{\min\{M,N\}}(1-\nu_i)^{+}<r\right\}.
\end{equation}
Using (\ref{outage2}), the probability of error can be upper bounded as follows:
\begin{eqnarray}\label{UB_outage}
\Pr(E_s)\leq {\rm Pr}(\pmb{\nu}\in\mathcal{B})+{\rm Pr}(E_s,\pmb{\nu}\in\overline{\mathcal{B}}).
\end{eqnarray}
The behavior of the first term at high SNR is $\rho^{-d^*_{\rm out}(r)}$. Averaging the second term over the channels in $\overline{\mathcal{B}}$ set, we obtain (see \cite{GCD}),
\begin{eqnarray}\label{PEs1}
\Pr(E_s)&\limi{\leq}&\rho^{-d^*_{\rm out}(r)}+\int_{\overline{\mathcal{B}}}f_{\pmb{\nu}}(\pmb{\nu})\Pr(E_s|\pmb{\nu})\;d\pmb{\nu}\cr
&\limi{\leq}&\rho^{-d^*_{\rm out}(r)},
\end{eqnarray}
where $f_{\pmb{\nu}}(\pmb{\nu})$ is the joint probability density function of $\pmb{\nu}$ given by (\ref{pdf1}).

\section*{Appendix II\\Proof of Theorem 3}
We consider an ensemble of $2MT$-dimensional random lattices $\{\Lambda_c\}$ with fundamental volume $V_c$ satisfying the Minkowski-Hlawka theorem (see \cite{GCD}, Theorem 1). The random lattice codebook is $\mathcal{C}(\Lambda,\pmb{u}_0,\mathcal{R})$, for some fixed translation vector $\pmb{u}_0$ and where $\mathcal{R}$ is the $2MT$-dimensional sphere of radius $\sqrt{MT}$ centered at the origin. The average probability of error (average over the channel and lattice ensemble) can be upper bounded as
\begin{equation}
\begin{split}
\bar{P}_e(\rho)&=\mathsf{E}_{\Lambda}\{P_e(\rho|\Lambda)\}\\
&\leq \mathsf{E}_{\Lambda}\{{\rm Pr}({\rm error},R_b(\rho)>R(\rho))\}+P_{\rm out}(\rho,b),
\end{split}
\end{equation}
where $P_e(\rho|\Lambda)$ is the probability of error for a given choice of $\Lambda$. Denote $0\leq\lambda_1\leq\cdots\leq \lambda_M$ the eigenvalues of $(\pmb{H}^c)^\mathsf{H}\pmb{H}^c$, and let $R=r\log\rho$. As shown in Section IV.B, by expressing the bias term $b$ as in (\ref{bias_SNR}), the achievable rate of lattice sequential decoding can be written as $R_b=\log\eta$, where $\eta=\phi\prod_{i=1}^M(1+\rho\lambda_i)^{\zeta_i}$. Now, define the asymptotic outage event $\mathcal{B}=\{\pmb{\beta}\in\mathbb{R}_+^{M}: \sum_{i=1}^{M}\zeta_i(1-\beta_i)^{+}<r\}$, where $\beta_i=-{\log\lambda_i/ \log\rho}$. Then, the first term in the RHS of the above upper bound can be expressed as
\begin{equation}\label{p0}
\begin{split}
\mathsf{E}_{\Lambda}\{{\rm Pr}({\rm error},R_b(\rho)>R(\rho))\}&\limi{=}\int_{\overline{\mathcal{B}}}f_{\pmb{\beta}}(\pmb{\beta})\mathsf{E}_{\Lambda}\{P_e(\rho|\pmb{\beta},\Lambda)\}\;d\pmb{\beta}\\
&\leq{\rm Pr}(|\pmb{e}'|^2>MT(1+\gamma))+\int_{\overline{\mathcal{B}}}f_{\pmb{\beta}}(\pmb{\beta}){\rm Pr}(\mathcal{A}|\pmb{\beta})\;d\pmb{\beta},
\end{split}
\end{equation}
where $\gamma>0$, and $f_{\pmb{\beta}}(\pmb{\beta})$ is the joint probability density function of $\pmb{\beta}$ which is asymptotically given by
\begin{equation}
f_{\pmb{\beta}}(\pmb{\beta})\limi{=}\exp\left(-\log(\rho)\sum\limits_{i=1}^{M}(2i-1+|N-M|)\beta_i\right).
\end{equation}

Consider here the Stack algorithm ($\delta=0$). In this case, the matrix $\tilde{\pmb{B}}$ provided in (\ref{newB}) can be expressed at high SNR as
$$\pmb{B}'=\left(1-b\rho^{-[\sum_{i=1}^M(1-\beta_i)^+-r]/M}\right)\pmb{B}.$$
Hence, at high SNR we have
\begin{equation}\label{p1}
\det(\tilde{\pmb{B}}^\mathsf{T}\tilde{\pmb{B}})\limi{=}\left(1-b\rho^{-[\sum_{i=1}^M(1-\beta_i)^+-r]/M}\right)\rho^{\sum_{i=1}^M(1-\beta_i)^+}.
\end{equation}
As $\rho\rightarrow\infty$, we can express $b$ (see (\ref{bias_SNR})) as
\begin{equation}\label{p2}
b\limi{=}{\rho^{\sum_{i=1}^M(1-\beta_i)^+/M}\over \eta^{1/M}}\left[1-\left({\eta\over \rho^{\sum_{i=1}^M(1-\beta_i)^+}}\right)^{1/2M}\right].
\end{equation}
Substituting (\ref{p2}) into (\ref{p1}), and by realizing that for all $R_b>R$ or equivalently $\eta\limi{>}\rho^r$, we can lower-bound (\ref{p1}) as $\det(\tilde{\pmb{B}}^\mathsf{T}\tilde{\pmb{B}})\geq\eta$.
Setting $\pmb{A}=\tilde{\pmb{B}}$ in Lemma 1, the ambiguity probability can be upper bounded as
\begin{equation}\label{p3}
\Pr(\mathcal{A}|\pmb{\beta})\limi{\leq}\exp(-T[\log\eta-r\log\rho]).
\end{equation}
It has been shown in \cite{GCD} that for $T\geq M+N-1$, the SNR exponent of ${\rm Pr}(|\pmb{e}'|^2>MT(1+\gamma))$ with respect to $\log\rho$ is larger than $d_0(r)>d_b(r)$. Substituting (\ref{p3}) in (\ref{p0}) we get (for $T\geq M+N-1$)
\begin{equation}\label{final_bound}
\begin{split}
\mathsf{E}_{\Lambda}\{{\rm Pr}({\rm error},&R_b(\rho)>R(\rho))\}\cr
&\limi{\leq}\int_{\overline{\mathcal{B}}}\exp\biggl(-\log(\rho)\sum\limits_{i=1}^{M}(2i-1+|N-M|)\beta_i+T\left[\sum\limits_{i=1}^{M}\zeta_i(1-\beta_i)^+-r\right]\biggr)
\;d\pmb{\beta}\cr
&\limi{=}\rho^{-d_b(r)}.
\end{split}
\end{equation}

\section*{Appendix III:\\ Proof of Lemma 1}
Without loss of generality, we assume that the all-zero lattice point was transmitted. Let
 \begin{equation}\label{phi1}
\phi'(\pmb{z}_1^k)=\begin{cases} 1, &\text{if $|{\pmb{e}'}_1^k-\pmb{R}_{kk}\pmb{z}_1^k|^2\leq bk+R_s^2$, $|{\pmb{e}'}_1^k|^2\leq R_s^2$;}\cr
                                                          0, &\text{otherwise.}\end{cases}
                                                          \end{equation}
where ${\pmb{e}'}_1^k$ is the last $k$ components of $\pmb{e}'=\pmb{Q}^\mathsf{T}\pmb{e}$, and $\pmb{Q}$ is the orthogonal matrix of the QR-decomposition of $\pmb{BG}$, and $R_s^2=MT(1+\log\rho)$. Given that $|\pmb{e}'|^2\leq R_s^2$, it must follow that $|{\pmb{e}'}_1^k|\leq R_s^2$, for all $1\leq k\leq m$. The total number of integer lattice points that satisfy (\ref{phi1}) can be upper bounded by
  \begin{equation}\label{L1}
  S_k\leq\sum\limits_{\pmb{z}_1^k\in\mathbb{Z}^k}\overline{\phi}(\pmb{z}_1^k).
  \end{equation}
  where
  \begin{equation}\label{phi2}
\overline{\phi}(\pmb{z}_1^k)=\begin{cases} 1, &\text{if $|{\pmb{e}'}_1^k-\pmb{R}_{kk}\pmb{z}_1^k|^2\leq bk+R_s^2$, $|{\pmb{e}'}_1^k|^2\leq bk+ R_s^2$;}\cr
                                                          0, &\text{otherwise.}\end{cases}
                                                          \end{equation}
 In general one can show that for any random vectors $\pmb{u}$ and $\pmb{v}$, and $r_s>0$, it holds$\{|\pmb{u}-\pmb{v}|^2\leq r_s^2,|\pmb{v}|^2\leq r_s^2\}\subseteq \{|\pmb{v}|^2\leq 4r_s^2\}$. Therefore, one can easily show that
 \begin{equation}
S_k\leq \sum\limits_{\pmb{z}_1^k\in\mathbb{Z}^k}\hat{\phi}(\pmb{z}_1^k),
 \end{equation}
 where
 \begin{equation}\label{phi2}
\hat{\phi}(\pmb{z}_1^k)=\begin{cases} 1, &\text{if $|\pmb{R}_{kk}\pmb{z}_1^k|^2\leq 4(bk+R_s^2)$;}\cr
                                                          0, &\text{otherwise.}\end{cases}
                                                          \end{equation}

We can further upper bound $S_k$ by introducing an auxiliary random variable that has a uniform distribution in the Voronoi region of the lattice $\Lambda(\pmb{R}_{kk})$. This can be done as follows:

Let
$$\tilde{\phi}(\pmb{x}_1^k+\pmb{u}_1^k)=\begin{cases} 1, & |\pmb{x}_1^k+\pmb{u}_1^k|^2\leq 7(bk+R_s^2)\\ 0, & \text{otherwise}\end{cases}$$
where $\pmb{u}_1^k$ is a random variable that is uniformly distributed in $\mathcal{V}_{\pmb{0}}(\pmb{R}_{kk})$ and independent of $\pmb{x}_1^k$. Then, assuming that there exists at least one lattice point $\pmb{x}_1^k\neq\pmb{0}$ inside the sphere, one can show that
$$S_k\leq \sum_{\pmb{x}_1^k\in \Lambda(\pmb{R}_{kk})}\tilde{\phi}(\pmb{x}_1^k+\pmb{u}_1^k)$$

The indicator function in (\ref{phi2}) can be rewritten as
\begin{align*}\hat{\phi}(\pmb{x}_1^k)&=\begin{cases} 1, &  |\pmb{x}_1^k|^2\leq 4(bk+R_s^2),\;|(\pmb{x}_1^k+\pmb{u}_1^k)-\pmb{u}_1^k|^2\leq 4(bk+R_s^2)\\ 0, & \text{otherwise}\end{cases}\\
&=\begin{cases} 1, & |\pmb{x}_1^k|^2\leq 4(bk+R_s^2),\; |\pmb{x}_1^k+\pmb{u}_1^k|^2\leq 4(bk+R_s^2)+2{\pmb{u}_1^k}^\mathsf{T}\pmb{x}_1^k+|\pmb{u}_1^k|^2\\ 0, & \text{otherwise}\end{cases}
\end{align*}
where $\pmb{u}_1^k$ is a uniform random variable in the fundamental region of the lattice $\Lambda(\pmb{R}_{kk})$. By noting that $|\pmb{u}_1^k|^2\leq (bk+R_s^2)$ [since $\pmb{u}_1^k\in\mathcal{V}_{\pmb{0}}(\Lambda(\pmb{R}_{kk}))$], and ${\pmb{u}_1^k}^\mathsf{T}\pmb{x}_1^k\leq |\pmb{u}_1^k||\pmb{x}_1^k|\leq (bk+R_s^2)$ (since $|\pmb{x}_1^k|\leq R_s$), we then have
$$\sum\limits_{\pmb{x}_1^k\in \Lambda(\pmb{R}_{kk})}\hat{\phi}(\pmb{x}_1^k)\leq \sum\limits_{\pmb{x}_1^k\in \Lambda(\pmb{R}_{kk})}\tilde{\phi}(\pmb{x}_1^k+\pmb{u}_1^k)$$
Equivalently, we have that
\begin{equation}\label{aaa}
S_k\leq \sum\limits_{\pmb{x}_1^k\in \Lambda(\pmb{R}_{kk})}\tilde{\phi}(\pmb{x}_1^k+\pmb{u}_1^k).
\end{equation}

Now, taking the average in both sides of (\ref{aaa}) over $\pmb{u}_1^k\in\mathcal{V}_{\pmb{0}}(\pmb{R}_{kk})$ we have (see Lemma 2 in \cite{Loe})
$$S_k\leq {V(\mathcal{S}_k(\sqrt{7(bk+R_s^2)}))\over V_f(\Lambda(\pmb{R}_{kk}))}$$

\begin{figure}[ht!]
\center
\includegraphics[width=3.5in]{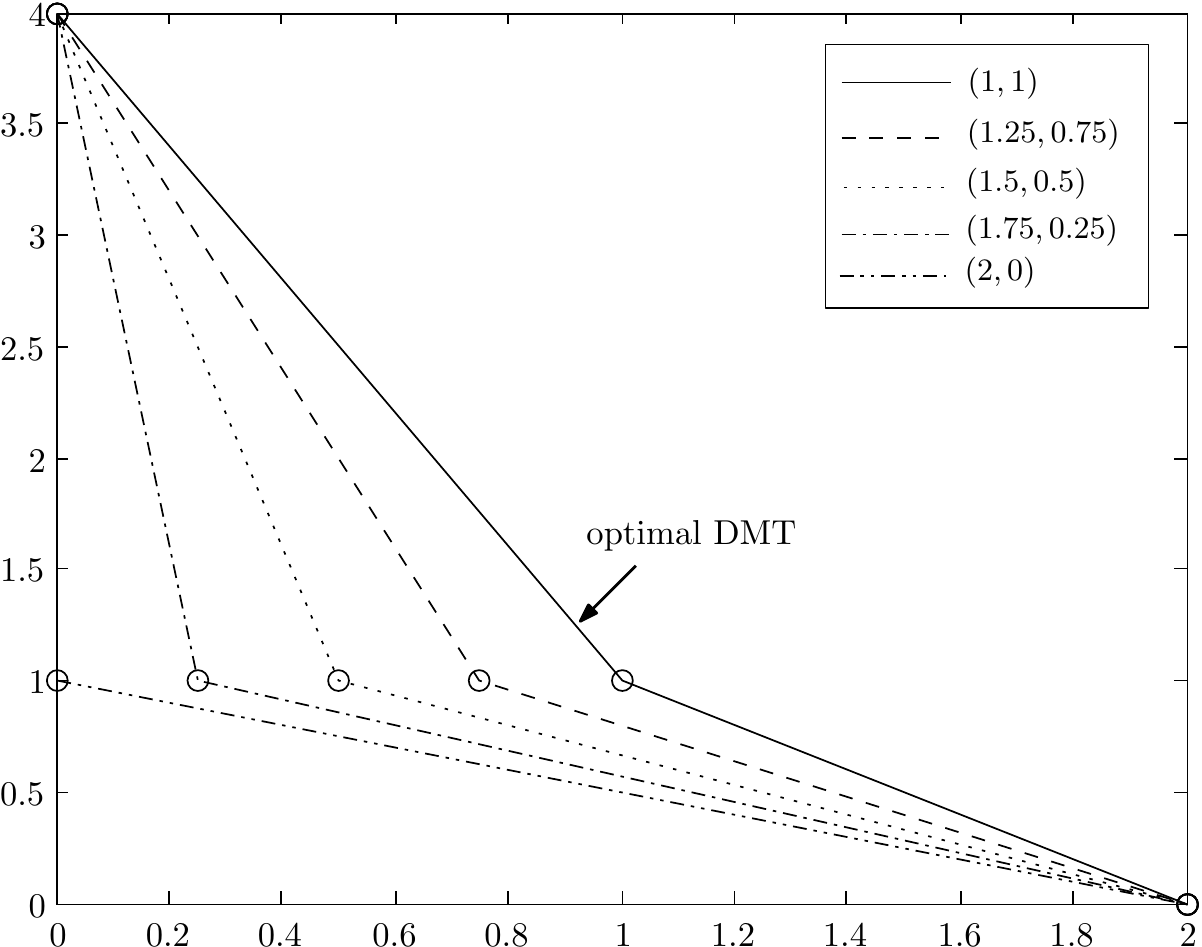}
\caption{DMT curves $d_b(r)$ achieved by lattice Fano/Stack sequential decoder for the case of 2$\times 2$ MIMO channel for different values of $(\zeta_1,\zeta_2)$.}
\end{figure}

\begin{figure}[ht!]
\centering

\subfigure[DMT curves correspond to \textit{Case} 1 in Theorem 4.]{
   \includegraphics[width =3.4in] {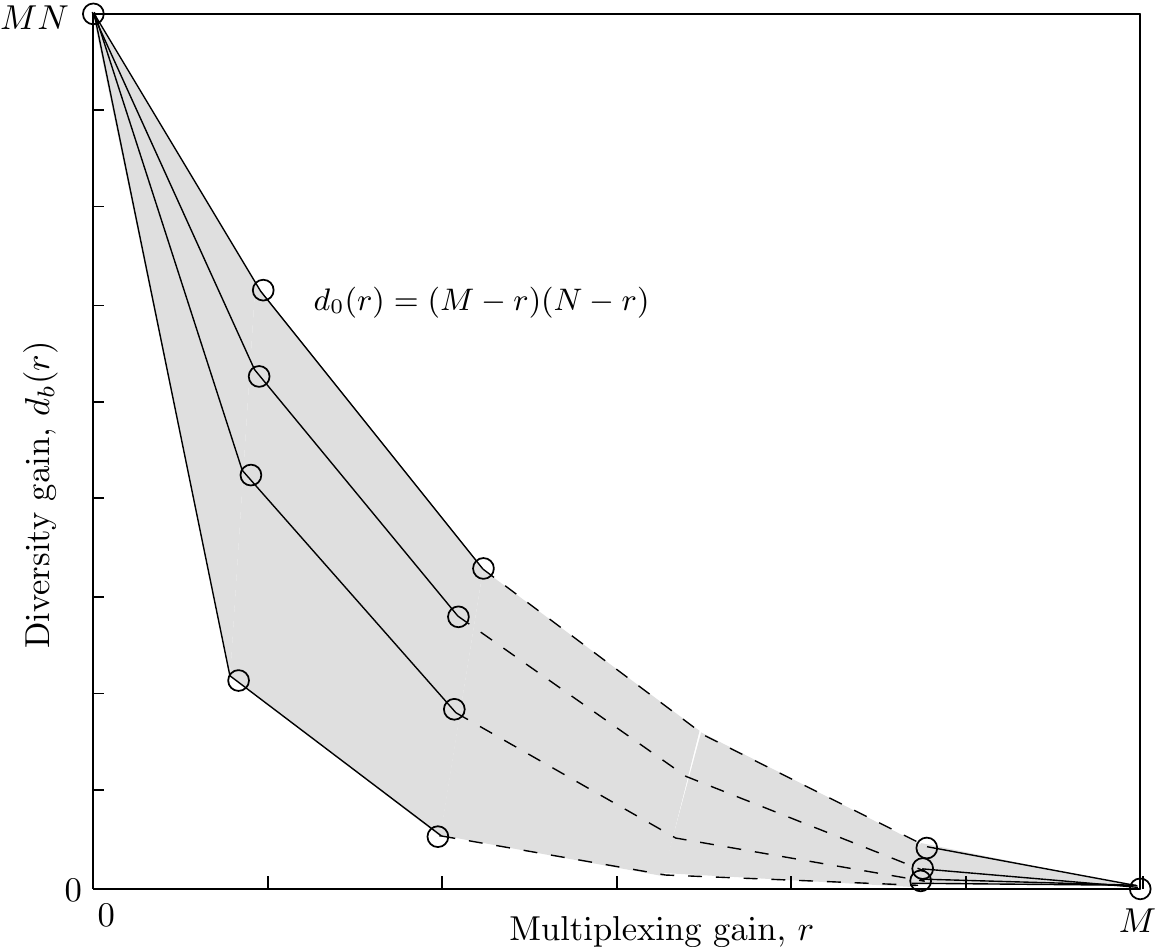}
   \label{fig:subfig1}
 }

 \subfigure[DMT curves correspond to \textit{Case} 2]{
   \includegraphics[width=3.4in] {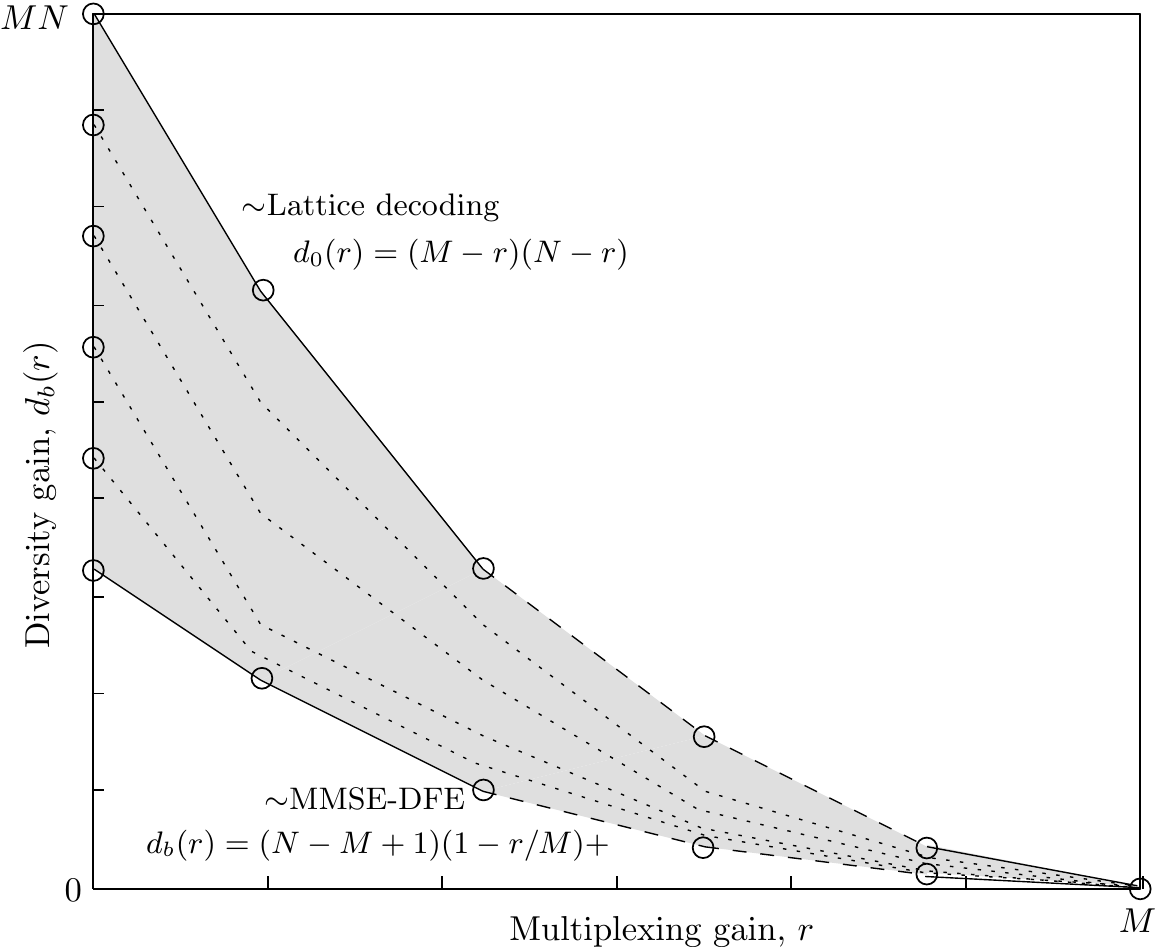}
   \label{fig:subfig2}
 }

\label{myfigure}
\caption{DMT curves $d_b(r)$ achieved by lattice Fano/Stack sequential decoder for different bias $b$.}
\end{figure}

\begin{figure}[ht!]
\center
\includegraphics[width=3.5in]{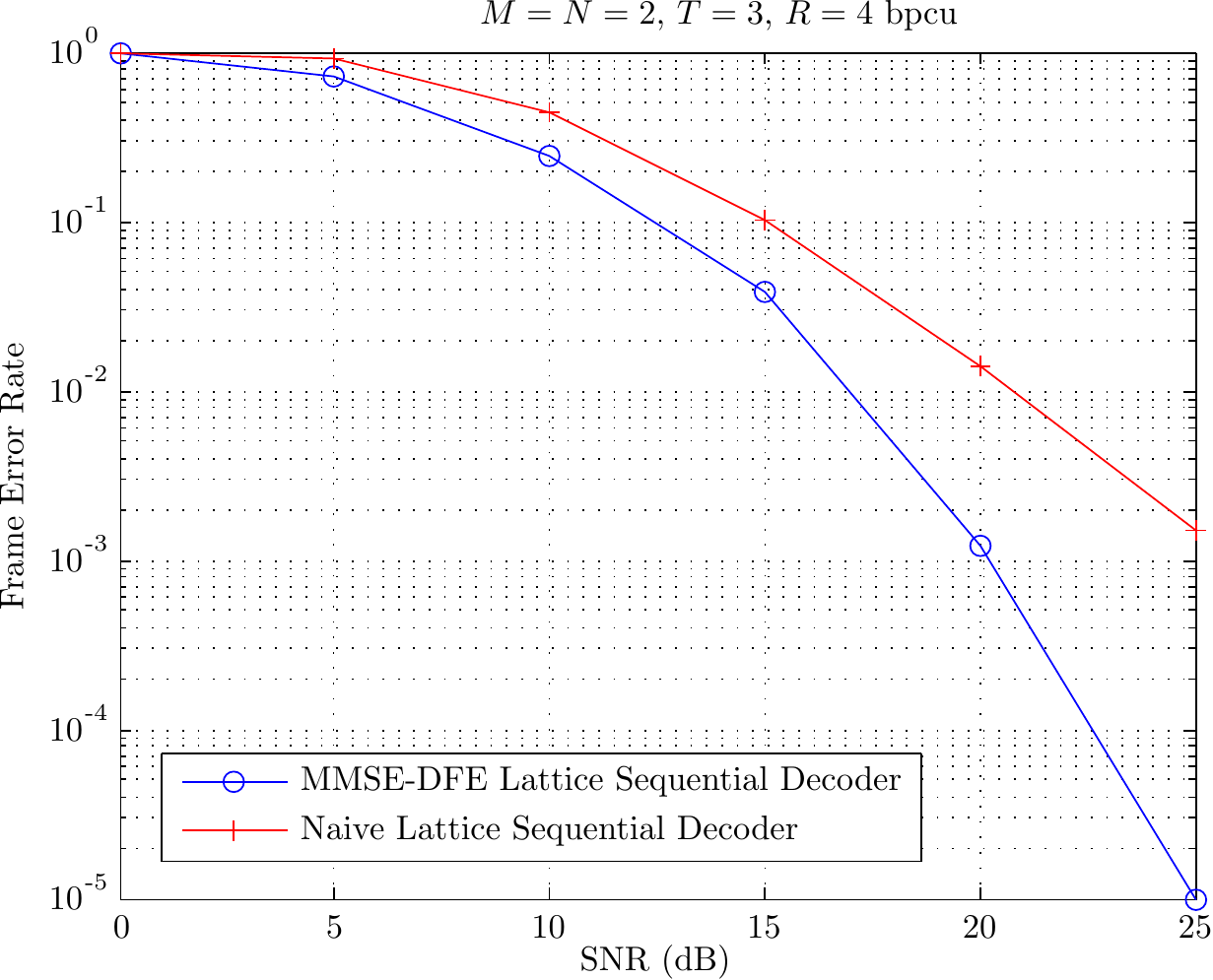}
\caption{Performance comparison between naive and MMSE-DFE lattice sequential decoding with $b=0.6$ for the case of $2\times 2$ LAST coded MIMO channel with $T=3$ and $R=4$ bpcu.}
\end{figure}

\begin{figure}[ht!]
\center
\includegraphics[width=3.5in]{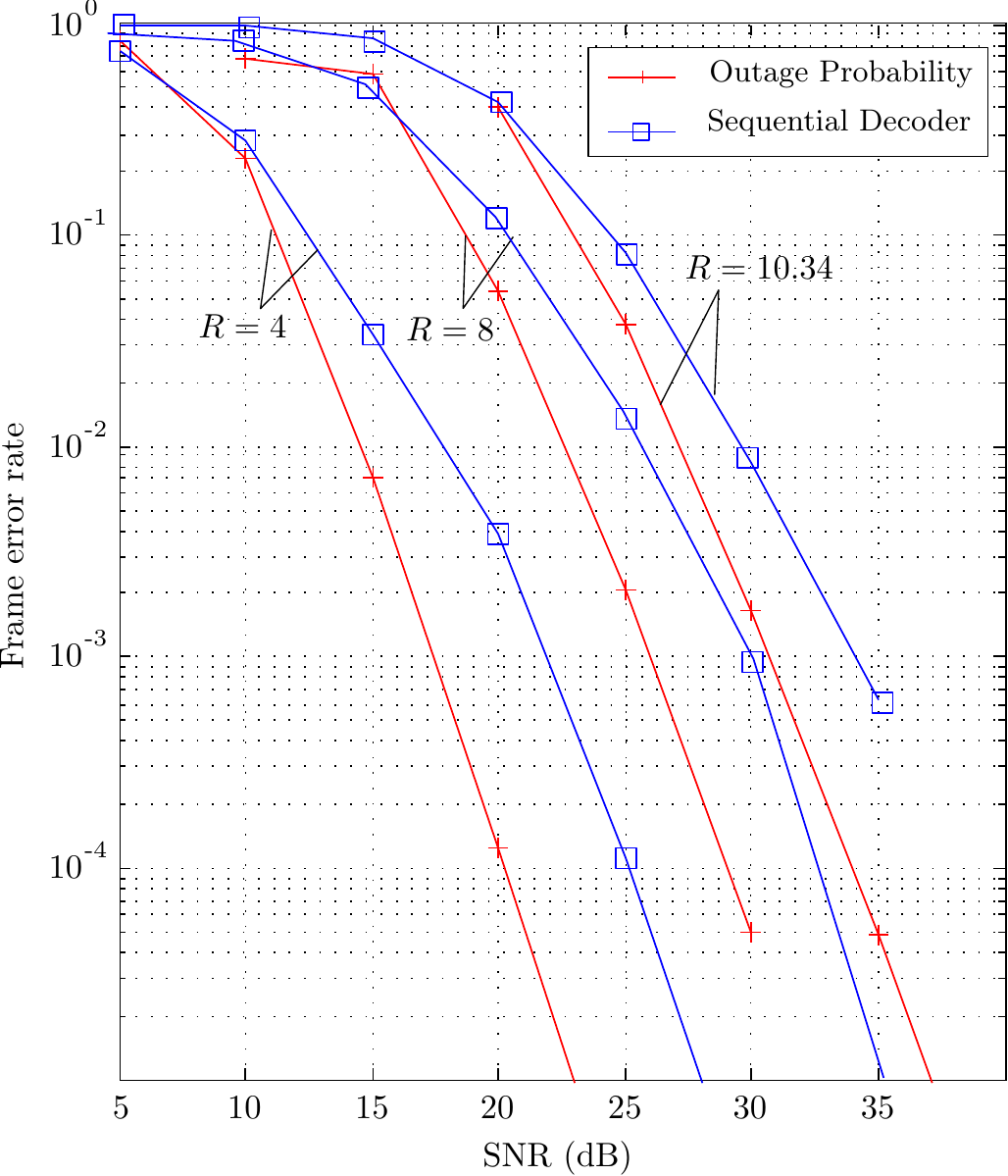}
\caption{Outage probability and error rate performance of lattice sequential decoding with $b=1$.}
\end{figure}

\begin{figure}[ht!]
\center
\includegraphics[width=3.5in]{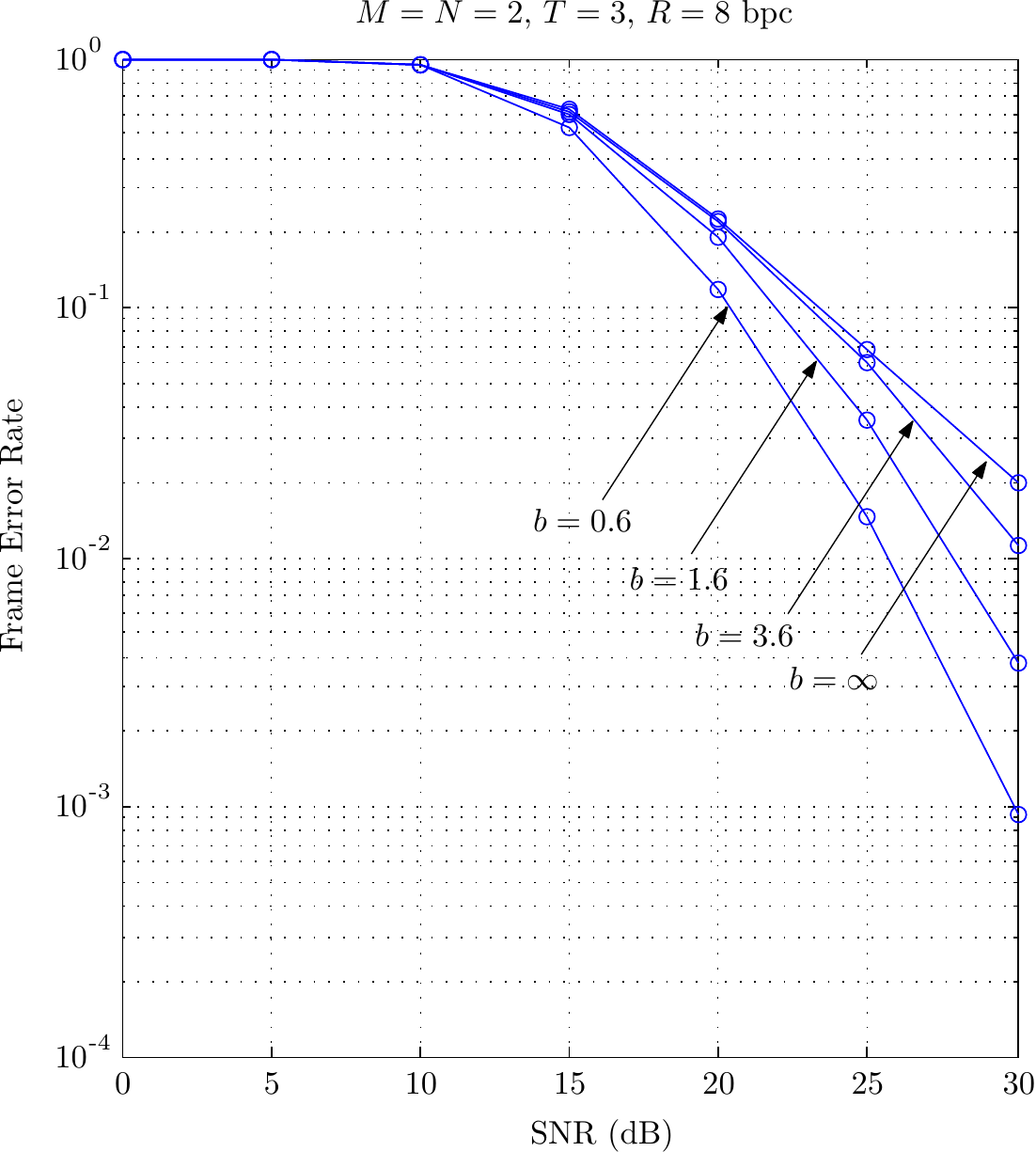}
\caption{Comparison of diversity order achieved by lattice sequential decoding for several values of $b$.}
\end{figure}

\begin{figure}[ht!]
\center
\includegraphics[width=3.5in]{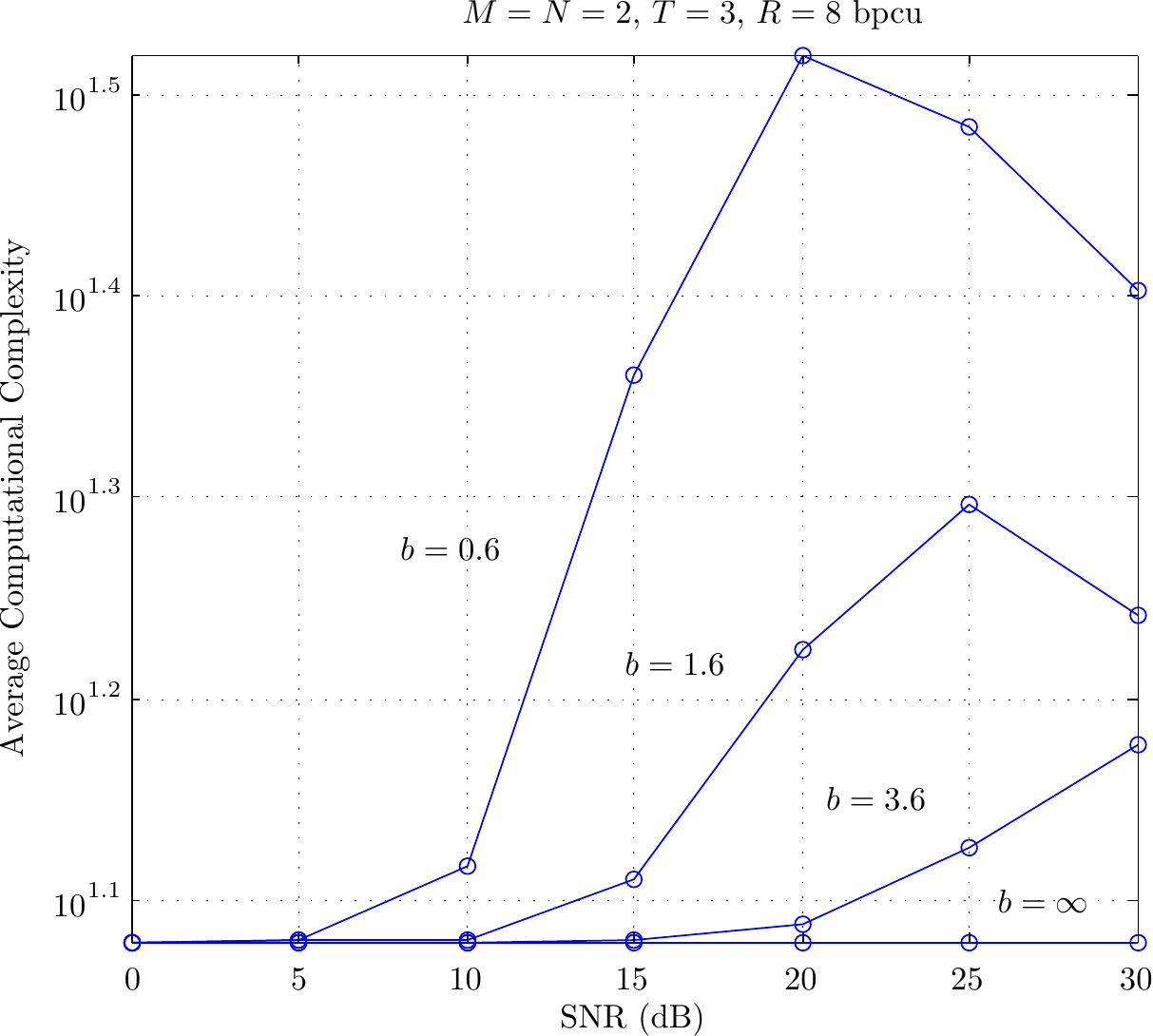}
\caption{Comparison of average computational complexity achieved by lattice sequential decoding for several values of $b$.}
\end{figure}

%\begin{figure}[ht!]
%\center
%\includegraphics[width=4in]{Aler_N1.pdf}
%\caption{$(a)$ Performance and $(b)$ complexity distribution achieved by the naive lattice sequential decoder ($b=0.6$) for the case of 2$\times 2$ LAST coded MIMO channel.. }
%\end{figure}
\begin{figure}[ht!]
\center
\includegraphics[width=4in]{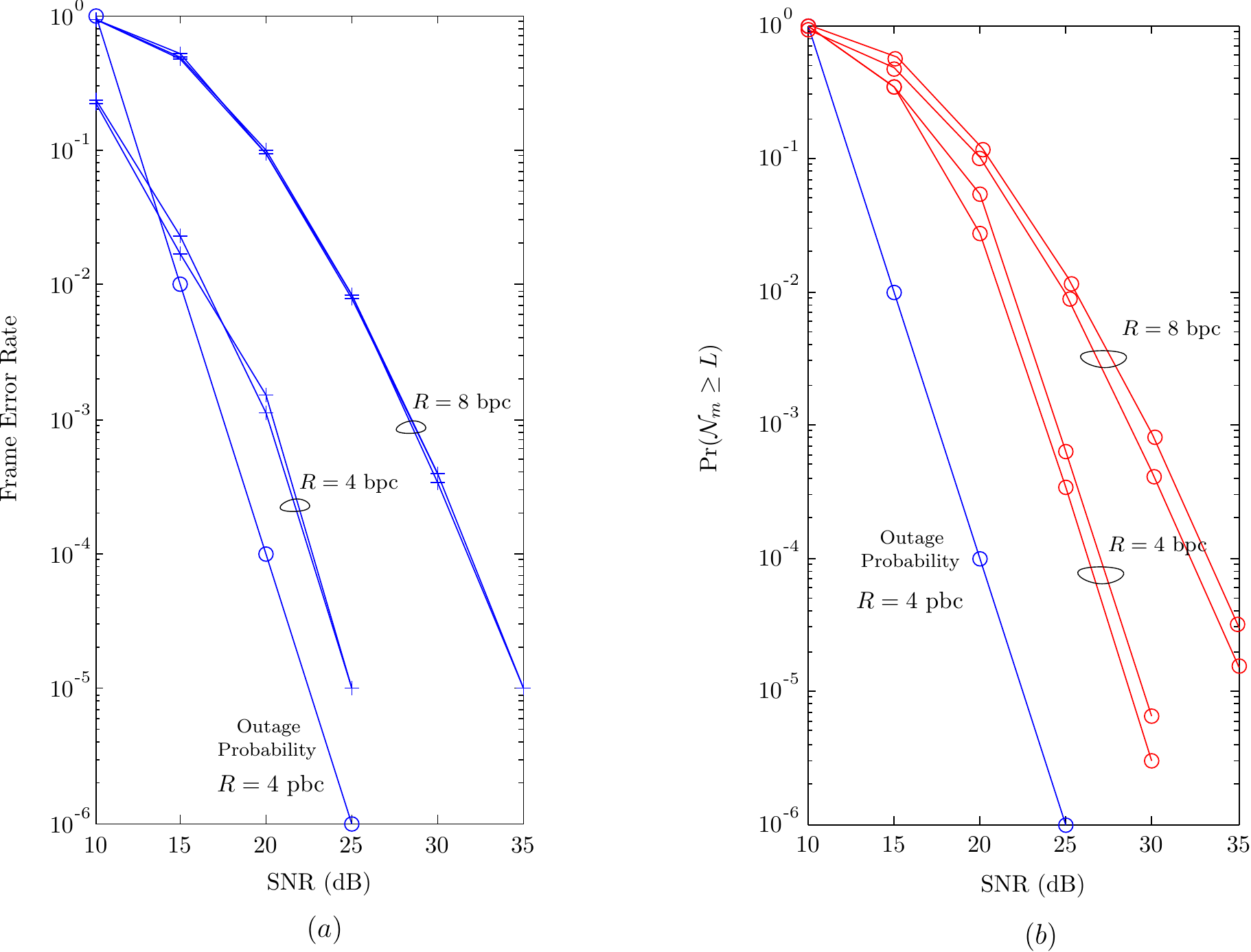}
\caption{$(a)$ Performance and $(b)$ complexity distribution achieved by the MMSE-DFE lattice sequential decoder ($b=0.6$) for the case of 2$\times 2$ LAST coded MIMO channel.}
\end{figure}

\begin{figure}[ht!]
\center
\includegraphics[width=5in]{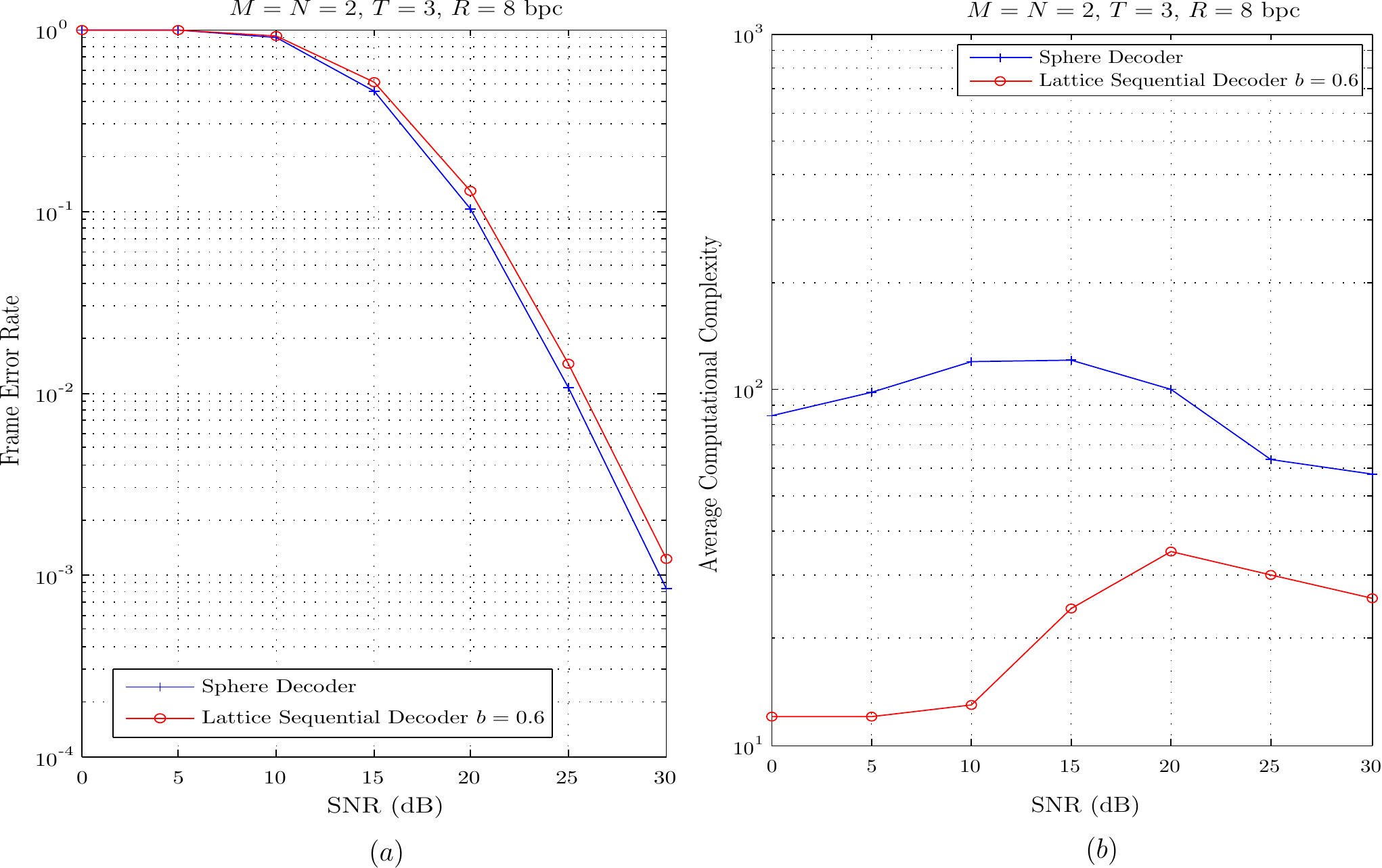}
\caption{$(a)$ Performance and $(b)$ average computational complexity comparison between sphere decoding and lattice sequential decoding for signal with dimension $m=12$.}
\end{figure}

\begin{figure}[ht!]
\center
\includegraphics[width=5in]{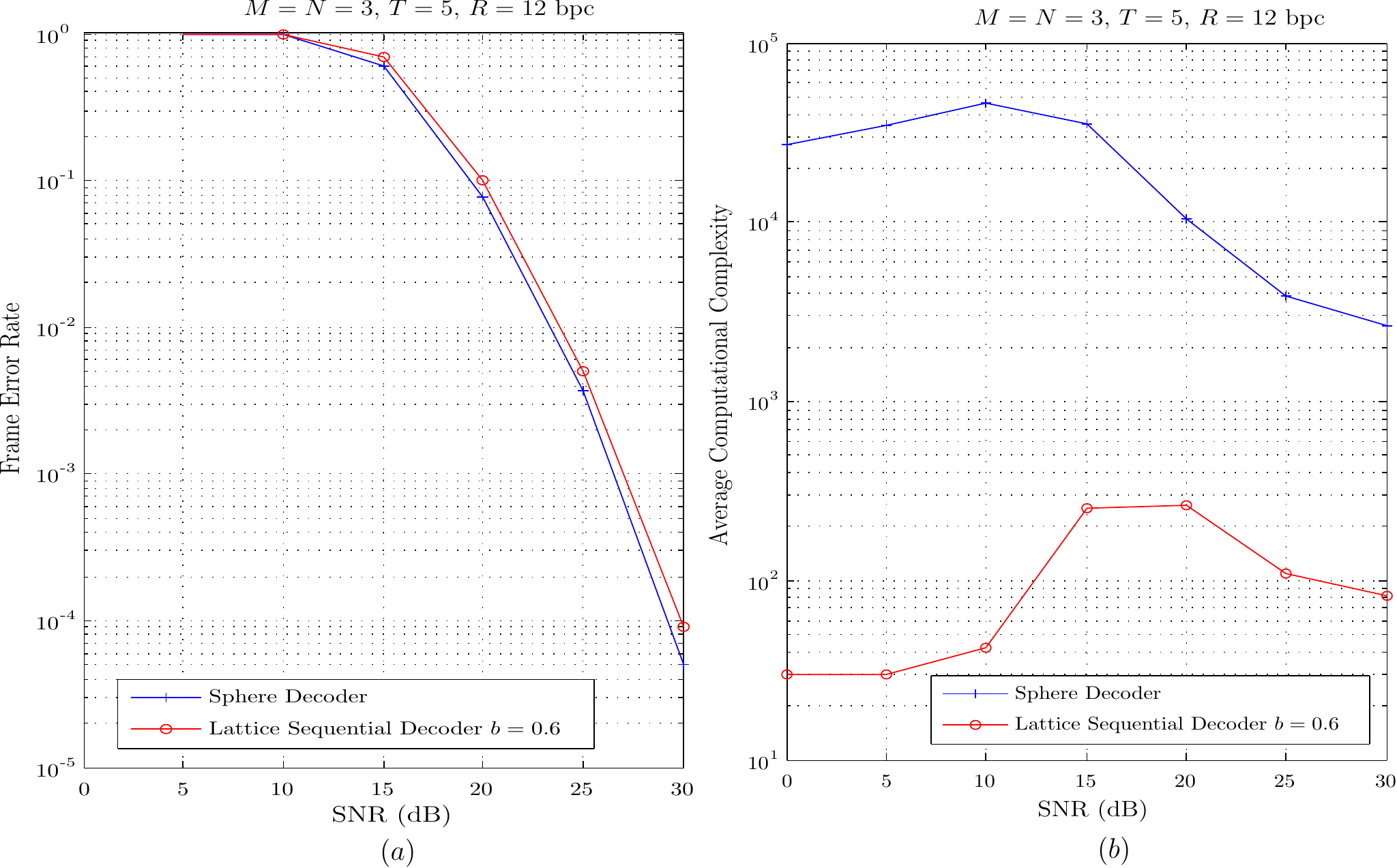}
\caption{$(a)$ Performance and $(b)$ average computational complexity comparison between sphere decoding and lattice sequential decoding for signal with dimension $m=30$.}
\end{figure}

%\begin{figure}[ht!]
%\center
%\includegraphics[width=4.5in]{Comp_naive2}
%\caption{The reduction in computational complexity achieved by the MMSE-DFE lattice sequential decoder compared to the naive one for several codeword lengths. The codeword lengths are selected so that the maximum diversity is achieved: $T\geq 1$ for naive decoding, and $T\geq 3$ for MMSE-DFE decoding. }
%\end{figure}

%\begin{figure}[ht!]
%\center
%\includegraphics[width=4.5in]{FER_naive1}
%\caption{Comparison of the performance achieved by the naive lattice sequential decoder with $b=0.3$ in a $2\times 2$ quasi-static MIMO channel, for different values of codeword length $T$.}
%\end{figure}

\begin{figure}[ht!]
\center
\includegraphics[width=4.5in]{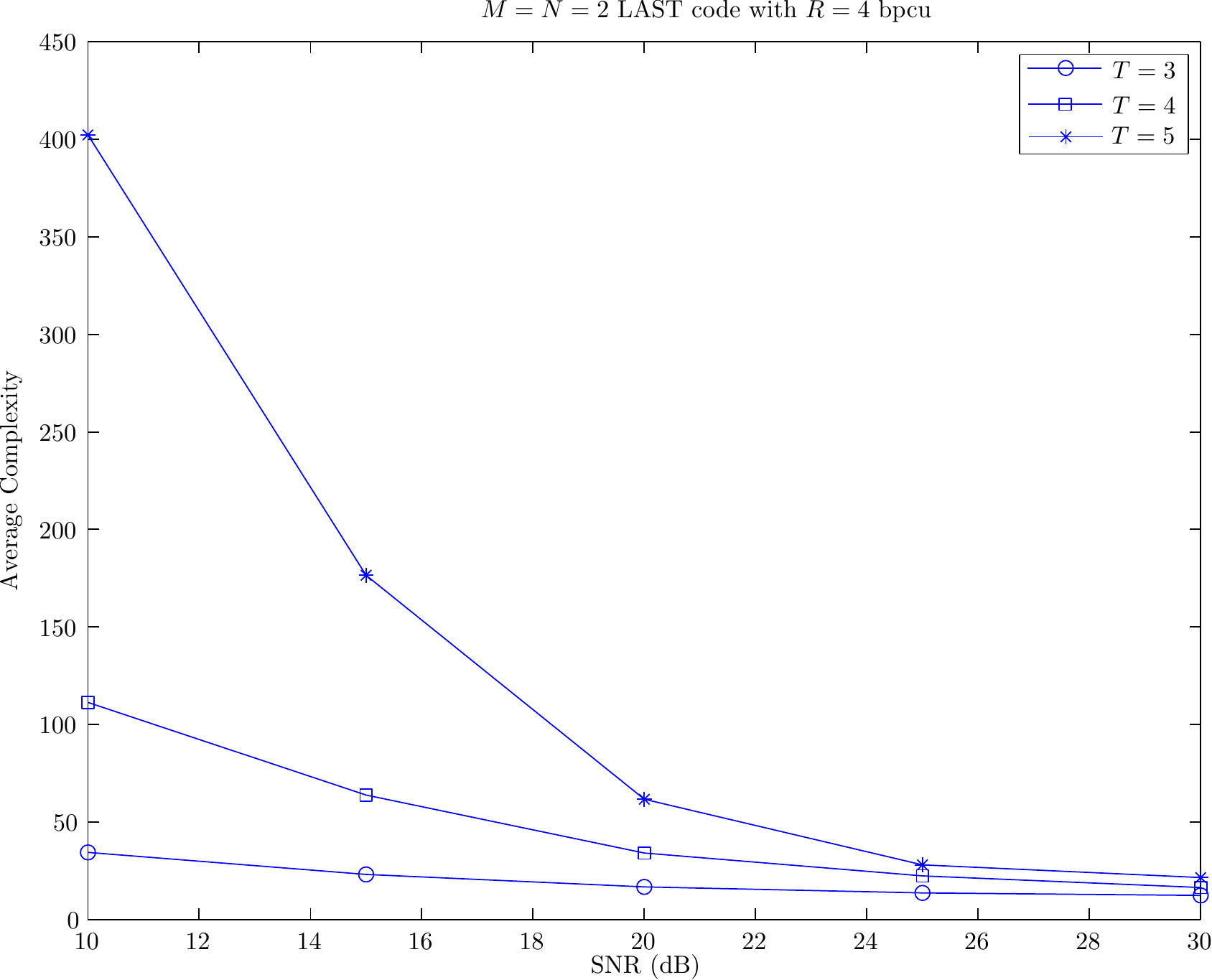}
\caption{The average computational complexity of the MMSE-DFE lattice Stack sequential decoder in a LAST coded $2\times 2$ MIMO system with different codeword lengths $T=3,4,5$ and fixed rate $R$. All curves decays quickly to $m=12,16,20$, respectively, as SNR increases.}
\end{figure}

\begin{figure}[ht!]
\center
\includegraphics[width=4.5in]{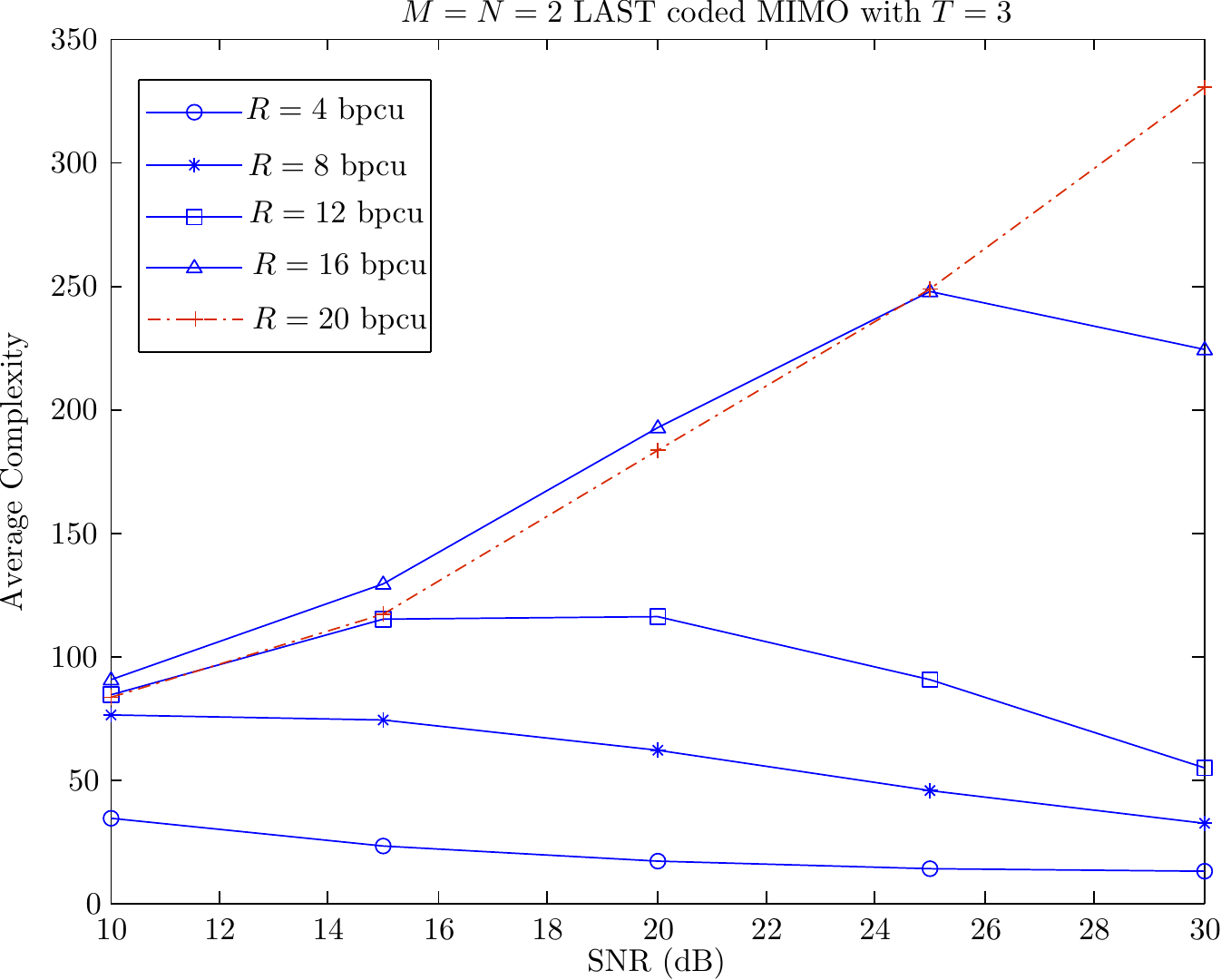}
\caption{Plots of the average complexity of the lattice sequential decoder for an optimal nested LAST coded $2\times 2$ MIMO system with different rates $R$ in bpcu.}
\end{figure}


\begin{thebibliography}{99}
\bibitem{Shannon1} C.~E.~Shannon, ``A Mathematical Theory of Communication'', \textit{Bell Syst. Techn. J.}, vol. 27, pp. 379-423, 623-656, July, October, 1948.

\bibitem{Conway} J.~H.~Conway and N.~J.~A.~Sloane, \textit{SpherePackings, Lattices, and Groups}, 3rd ed. Springer Verlag NewYork, 1999.

\bibitem{UR} R.~Urbanke, and B. Rimoldi, ``Lattice Codes Can Achieve Capacity on the AWGN Channel'', \textit{IEEE Trans.~on Inform.~Theory}, vol.~44, pp.~273-278, Jan. 1998.

\bibitem{deBuda} R.~deBuda, ``The upper bound of a new near-optimal code'', \textit{IEEE Trans. Inform. Theory}, vol.~IT-21, pp.~441-445, July 1975.

\bibitem{Polytrev}G.~ Polyterv, ``On coding without restrictions for the AWGN channel,'' \textit{IEEE Trans.~Inform.~Theory}, vol.~40, pp.~409-417, 1994.

\bibitem{Loe} H.~Loeliger, ``Averaging Bounds for Lattices and Linear Codes'', \textit{IEEE Trans.~Inform.~Theory}, vol.~43, no.~6, pp.~1767-11773, Nov.~1997.

\bibitem{Forney_MMSE} G. D. Forney Jr., ``On the role of MMSE estimation in approaching the information-theoretic limits of linear Gaussian channels: Shannon meets Wiener,'' in \textit{Proc. 41th Annu. Allerton Conf. Communications, Control, and Computing}, Monticello, IL, Oct. 2003.

\bibitem{MIMOARQ} H. El Gamal, G. Caire, M. Damen, ``The MIMO ARQ Channel: Diversity-Multiplexing-Delay Tradeoff'', \textit{IEEE Trans.~Inform.~Theory}, vol.~52, no.~8, August 2006.

\bibitem{CLSP} E.~Agrell, T.~Eriksson, A.~Vardy, and K.~Zeger, ``Closest point search in lattices,'' \textit{IEEE Trans.~Inform.~Theory}, vol.~48, no.~8, pp.~2201Ð2214, Aug.~2002.

\bibitem{Fano} R. M. Fano, ``A heuristic discussion of probabilistic decoding'', \textit{IEEE Trans. Inform. Theory}, vol.-9, pp~ 64-73, Apr.~1963.

\bibitem{Stack} F. ~Jelinek, ``A fast sequential decoding algorithm using a stack'', \textit{IBM J.~Res.~Dev.~}, 13:675-685,1969.

\bibitem{Massey} J.~L. Massey, ``Variable-length codes and the Fano metric'', \textit{IEEE Trans. Inform. Theory}, vol.-18, pp~ 196-198, 1972.

\bibitem{JZ} R. Johannesson and K. Zigangirov, \textit{Fundamentals of Convolutional Coding}. Piscataway, NJ: IEEE Press, 1999.

\bibitem{TVZ} V. Tarokh, A. Vardy, and K. Zeger, ``Sequential decoding of lattice codes,'' preprint, 1996.

\bibitem{FTC} G.D. Forney Jr., M.~D. Trott, and S.~Chung, ``Sphere-Bound-Achieving Coset Codes and Multilevel Coset Codes'', \textit{IEEE Trans. Inform. Theory}, vol.-46, pp~820-850, May 2000.

\bibitem{ONM} O.~Shalvi, N.~Sommer, and M.~Feder, ``Signal Codes: Convolutional Lattice Codes'', \textit{IEEE Trans. Inform. Theory}, vol.-57, no.~8, pp~5203-5226, Aug. 2011.

\bibitem{NMO} N.~Sommer, M.~Feder, and O.~Shalvi, ``Closest point search in lattices using sequential decoding'', \textit{IEEE Int. Symp. Inform. Theory}, p.~1053-1057, Adelaide, SA, Sept. 2005.

\bibitem{MGDC} A. Murugan, H. El Gamal, M. O. Damen and G. Caire, ``A unified framework for tree search decoding: rediscovering the sequential decoder'', \textit{IEEE Trans. Inform.~ Theory}, vol.~52, no.~3, March 2006.

\bibitem{YW} H. Yao and G. W. Wornell, ``Lattice-reduction-aided detectors for MIMO communication systems,'' in \textit{Proc. IEEE Global Communications Conference (GLOBECOM)}, vol. 1, pp. 424–428 Nov. 2002.

\bibitem{EJ} J. Jald\'{e}n, P. Elia, ``DMT Optimality of LR-Aided Linear Decoders for a General Class of Channels, Lattice Designs, and System Models'', \textit{IEEE Trans. on Inform. Theory}, vol. 56, no. 10, Oct. 2010.

\bibitem{GCD} H.~El Gamal, G.~Caire, M.~O.~Damen, ``Lattice coding and decoding achieve the optimal diversity-multiplexing tradeoff of MIMO channels'', \textit{IEEE Trans.~Inform.~Theory}, vol.~50, no.~6, pp.~968-985, June 2004.

\bibitem{DCB} M.~O.~Damen, A.~Chkeif, and J.~-C.~Belfiore, ``Lattice codes decoder for space-time codes,'' \textit{IEEE Commun.~Lett.~}, vol. 4, no.~ 5, pp.~161Ð163, May 2000.

\bibitem{CLPS} E.~Agrell, T.~Eriksson, A.~Vardy, and K.~Zeger, ``Closest point search in lattices,'' \textit{IEEE Trans. Inform.~Theory}, vol.~48, no.~8, pp.~2201Ð2214, Aug.~2002.

\bibitem{HV} B. Hassibi and H. Vikalo, ``On the sphere decoding algorithm: Part I, the expected complexity'', \textit{IEEE Trans. on Signal Processing}, vol 53, no 8, pages 2806-2818, Aug 2005.

\bibitem{HV2} B. Hassibi and H. Vikalo, ``On the sphere decoding algorithm: Part II, generalization, second-order statistics, and applications to communication'', \textit{IEEE Trans. Signal Processing}, vol 53, no 8, pages 2819-2834, Aug 2005.

\bibitem{DGC} M.~O.~Damen, H.~El Gamal, and G.~Caire, ``On maximum-likelihood detection and the search for the closest lattice point,'' \textit{IEEE Trans.~ Inform.~Theory}, vol.~49, no.~10, pp.~2389Ð2401, Oct.~2003.

\bibitem{ZT} Zheng and D.~Tse, ``Diversity and multiplexing: A fundamental tradeoff in multiple antenna channels,'' \textit{IEEE Trans.~Inform.~Theory}, vol.~49, no.~5, pp.~1073-1096, May 2003.

\bibitem{Fosch} G.~J.~Foschini, ``Layered space-time architecture for wireless communication in a fading environment when using multi-elements antenna,'' \textit{Bell Labs Tech.~J.~,} vol.~1, no.~2, pp.~41-59, 1996.

\bibitem{MMSE1} Y. Jiang, M. K. Varanasi, and J. Li, ``Performance analysis of ZF and MMSE equalizers for MIMO systems: A closer study in high SNR regime,'' \textit{IEEE Trans. Inform. Theory},

\bibitem{MMSE2} K. Kumar, G. Caire, and A. Moustakas, ``Asymptotic performance of linear receivers in MIMO fading channels,'' \textit{IEEE Trans. Inform. Theory}, vol.~55, no.~10, pp.~4398-4418, Oct.~2009.


\bibitem{Tarokh1} V.~Tarokh, N.~Seshadri, and A.~Calderbank, ``Space-time codes for high data rate wireless communications: Performance criterion and code ocnstruction,'' \textit{IEEE Trans.~Inform.~Theory}, vol.~44, pp.~744-756, Mar.~1998.

\bibitem{Hassibi} B.~Hassibi and B.~M.~Hochwald, ``High-rate codes that are linear in space and time,'' \textit{IEEE Trans.~Inform.~Theory}, vol.~48, no.~7, pp.~1804-24., Jul.~2002.


\bibitem{Jacobs} I.~M.~Jacobs, and E.~R.~Berlekamp, ``A lower bound to the distribution of computation for sequential decoding'', \textit{IEEE Trans. Inform.~Theory}, vol.~IT-13, pp~ 167-174, 1976.

\bibitem{JB} J. Jald\'{e}n, B. Ottersten, ``On the complexity of sphere decoding in digital communication'', \textit{IEEE Trans. Signal Processing}, vol. 53, no. 4, pp. 1474-1484, Apr. 2005.

\bibitem{SJ} D. Seethaler, J. Jald\'{e}n, C. Studer, and H. Bolcskei, ``On the complexity distribution of sphere-decoding'', \textit{IEEE Trans. Inform. Theory}, Dec. 2009, submitted.

\bibitem{EZ1} U.~Erez and R.~Zamir, ``Lattice coding can achieve $1/2\log(1+snr)$ on the AWGN channel using nested codes'', \textit{IEEE Trans.~Inform.~Theory}, vol.~50, no.~10, pp.~2293-2314, Oct.~2004.

\bibitem{EZ2} U.~Erez, S.~Litsyn, and R.~Zamir, ``Lattices which are good for (almost) everything,'' \textit{IEEE Trans.~Inform.~Theory}, vol.~51, no.~10, pp.~3401-3416, Oct.~2005.


\bibitem{HCM} H. El Gamal, G. Caire, M. Damen, ``The MIMO ARQ Channel: Diversity-Multiplexing-Delay Tradeoff'', \textit{IEEE Trans. Inform. Theory'}, vol. 52, no. 8, August 2006.

\bibitem{WD2} W. Abediseid, M. O. Damen, ``Time-Out Lattice Sequential Decoding for the MIMO ARQ Channel'', submitted to \textit{IEEE Trans. on Wireless Comm.}, 2010.

\bibitem{GW} P.~M.~Gruber, and J.~M.~Wills, Eds., \textit{Handbook of Convex Geometry}, vol.~B, North Holland, Amsterdam: Elsevier, 1993.



\end{thebibliography}
\end{document}